\pdfoutput=1
\documentclass[11pt,a4paper]{article}
\usepackage{bm}

\makeatletter
\newsavebox{\@brx}
\newcommand{\llangle}[1][]{\savebox{\@brx}{\(\m@th{#1\langle}\)}\mathopen{\copy\@brx\kern-0.5\wd\@brx\usebox{\@brx}}}
\newcommand{\rrangle}[1][]{\savebox{\@brx}{\(\m@th{#1\rangle}\)}\mathclose{\copy\@brx\kern-0.5\wd\@brx\usebox{\@brx}}}
\makeatother

\usepackage{lmodern}

\newcommand{\erf}[1]{\ensuremath{\mathrm{erf}\inp{#1}}}
\newcommand{\gaussian}{\mathcal{N}}
\newcommand{\dtv}{\ensuremath{d_{\mathrm{TV}}}}

\newcommand{\len}[1]{\mathrm{len}\inp{#1}}

\newcommand{\multnormtv}{\textsc{MultGaussianTV}}
\newcommand{\disprodtv}{\textsc{DisProdTV}}

\newcommand{\vct}[1]{\boldsymbol{#1}}
\newcommand{\Ex}{\mathbb{E}}

\newcommand{\Rank}{\textbf{Rank}}
\newcommand{\Range}{\textnormal{\textbf{Range}}}

\usepackage{esdiff}

\usepackage[round]{natbib}
\renewcommand{\bibname}{References}

\newcommand{\F}{\ensuremath{\mathcal{F}}} \newcommand{\cI}{\ensuremath{\mathcal{I}}} \newcommand{\tR}{\ensuremath{\tilde{R}}} \newcommand{\Rpos}{\ensuremath{\mathbb{R}_{\geq 0}}} \newcommand{\ratio}[2]{\ensuremath{{#1\|#2}}}
\newcommand{\indep}{\ensuremath{\mathbin{\underset{\textup{ind}}{\circ}}}}
\newcommand{\cdf}{\textup{c.d.f.}}
\newcommand{\ext}{\operatorname{ext}}
\newcommand{\abss}[1]{\ensuremath{|#1|}}

 \usepackage{psmacros}
\begin{document}

\title{Approximating the Total Variation Distance between Gaussians}
\author{Arnab Bhattacharyya\footnote{The University of Warwick} \and Weiming Feng\footnote{The University of Hong Kong} \and Piyush Srivastava \footnote{Tata Institute of Fundamental Research}}
\date{}

\maketitle

\begin{abstract}
    The total variation distance is a metric of central importance in statistics
    and probability theory.  However, somewhat surprisingly, questions about
    computing it \emph{algorithmically} appear not to have been systematically
    studied until very recently.  In this paper, we contribute to this line of
    work by studying this question in the important special case of multivariate
    Gaussians.  More formally, we consider the problem of approximating the
    total variation distance between two multivariate Gaussians to within an
    $\epsilon$-relative error.  Previous works achieved a \emph{fixed} constant
    relative error approximation via closed-form formulas. In this work, we give
    algorithms that given any two $n$-dimensional Gaussians $D_1,D_2$, and any
    error bound $\epsilon > 0$, approximate the total variation distance
    $D \defeq d_{TV}(D_1,D_2)$ to $\epsilon$-relative accuracy in
    $\mathrm{poly}(n,\frac{1}{\epsilon},\log \frac{1}{D})$ operations.  The main technical
    tool in our work is a reduction that helps us extend the recent progress on
    computing the TV-distance between \emph{discrete} random variables to our
    continuous setting.
\end{abstract}

\section{Introduction} 

A fundamental problem in applied probability is to understand how different two
probability distributions are from each other.  Depending upon the applications,
various measures of divergence are used for this purpose; one of the most
important of these is the \emph{total variation (TV) distance}.  Given two
probability distributions $P$ and $Q$ on the same set $\Omega$, the total
variation distance $\dtv(P, Q)$ is the largest possible difference, over all
\emph{tests} $S \subset \Omega$, between the probabilities $P(S)$ and $Q(S)$ of the
test ``succeeding'' on a sample drawn from $P$ and $Q$ respectively.

A natural question then is: given the specifications of $P$ and $Q$, can one
compute $\dtv(P, Q)$?  A standard alternative expression for the TV distance in
terms of the \(\ell_1\)-norm shows that when the sample space $\Omega$ is \emph{finite},
$\dtv(P, Q)$ can be computed exactly in $O(\abs{\Omega})$ arithmetic operations,
given the values of $P(\omega)$ and $Q(\omega)$ for each $\omega \in \Omega$.  However, even in the
slightly more general (and much more important!) setting of \emph{product
  distributions}, the problem becomes much more interesting.  Now, we are given
distributions $(P_i, Q_i)$ on the same finite set $\Omega$, for
$1 \leq i \leq n$, and we wish to compute the TV distance between the product
distributions $P \defeq P_1 \times \dots \times P_{n}$ and
$Q \defeq Q_1 \times \dots \times Q_{n}$ on $\Omega^n$.  The obvious algorithm from above now
has a prohibitive, exponentially growing cost of $\Theta(\abs{\Omega}^n)$.

Somewhat surprisingly, it seems that the computational complexity of computing
the TV distance, even in such simple settings, had not been systematically
studied until recently.  \cite{BGMMPV23} provided evidence that even in the
setting of product distributions, we probably cannot do much better than the
above obvious algorithm if we want to compute $\dtv$ exactly.  They did so by
proving a \#P-hardness result: a polynomial time algorithm for the above problem
would imply a polynomial time algorithm for counting satisfying assignments to
Boolean formulas.  We therefore turn to \emph{approximation}, and require
instead that given a further input $\epsilon \in (0, 1)$, the algorithm should output a
number $Z$ such that $\abs{Z - \dtv(P, Q)} \leq \epsilon\dtv(P, Q)$ (i.e., the algorithm
makes small relative error), in time that is polynomial in
$n, \abs{\Omega}, 1/\epsilon$ and $\log(1/\dtv(P, Q))$. We note here that the dependence on
$\log(1/\dtv(P, Q))$ is reasonable. First, we may assume that
$\dtv(P,Q) \neq 0$, since otherwise, the problem is trivial.  Further, if the final
answer is a binary number, then the output bit length itself is
$\Omega(\log(1/\dtv(P, Q)))$. In addition, if all the input numbers are rationals
specified as ratios of binary integers, then $\log(1/\dtv(P, Q))$ is at
most a polynomial in the input size.\footnote{However, if we measure only the
  number of \emph{arithmetic} operations, rather than the number of \emph{bit}
  operations, then this dependence on \(\log(1/\dtv(P, Q))\) is not necessarily
  essential.  We comment further on this point later in the paper.}
We refer to
this problem as \disprodtv; see \cref{prob-dis-prod-tv} for a formal
description.  \cite{BGMMPV23} gave algorithms for \disprodtv{} in certain
special cases.  An elegant randomized algorithm for \disprodtv{} that covers all
cases was then given by \cite{FGJW23}.  More recently, \cite{FengLL24} gave a
deterministic polynomial time algorithm for \disprodtv.

\subsection{Our Contribution} In this paper, we consider the problem of
computing the TV distance between $P$ and $Q$ when they are Gaussian
distributions in $\R^n$, specified by their mean vectors and covariance
matrices.  Towards this goal, we make the following contributions.
 
\paragraph{A General (and Simpler) Analysis } \cite{FengLL24} developed a
framework for studying TV distance estimation between products of discrete
distributions on finite sets.  We simplify and also extend this framework to the
setting of general product distributions (including those without a probability
density function).  Our framework allows us to bypass one of the main
innovations of \cite{FengLL24}, which they referred to as the $\Delta_{MTV}$
distance, by recasting the analysis in terms of a more direct quantity that we
call the \emph{extension distance}.  As an application, we give an arguably
simpler proof of correctness of the algorithm of \cite{FengLL24} in
\cref{sec:tv-dist-estim}.
\paragraph{TV Distance between Gaussians} We then apply our framework to the
problem of estimating the TV distance between $P$ and $Q$ when $P$ and $Q$ are
arbitrary Gaussian distributions in $\R^n$.  Our first step is to use the linear
algebraic properties of Gaussians to reduce the problem to the case where $P$
and $Q$ are both products of one-dimensional Gaussians.  Assuming that linear
algebraic operations such as eigen-decomposition of symmetric matrices can be
performed exactly, this part of the reduction is exact; however, we comment on
the bit-complexity issues arising in this part in \Cref{sec:algor-with-appr}.
After this reduction, we are \emph{almost} in the setting of prior work
\citep{FGJW23,FengLL24} that studied the problem of estimating the TV distance
between product distributions, \emph{except} that the time complexity in these
works has a polynomial dependence on the size of the state space $\Omega$ of
each component of the product.  For us, however, this set $\Omega$ is the
uncountable set $\R$ of all real numbers.  We thus need to perform a careful
discretization step to bring the problem to the setting of prior work.  Happily,
it turns out that this discretization can be analyzed using our extension
(described above) of the framework of \cite{FengLL24}.  Combining everything, we
obtain an algorithm for approximating TV distance between arbitrary Gaussian
distributions $P$ and $Q$ in $\R^n$ to a relative accuracy of $\pm\epsilon$ in
time polynomial in $n,1/\epsilon$ and $\log(1/\dtv(P, Q))$, in a computational
model where eigenvalues and eigenvectors of symmetric matrices are assumed to be
exactly computable.  Below, we give the statements of our results regarding TV
distance between Gaussians; the proofs of these results are given in
\cref{sec:alg}.  The problem we are interested in is the following.
\begin{problem}[\textbf{\multnormtv{}}]\label{prob:TV}
  \textbf{INPUT}: Two $n$-dimensional mean vectors
  $\vct{\mu}_1, \vct{\mu}_1 \in \mathbb{R}^n$ and two $n \times n$ positive semi-definite covariance
  matrices $\Sigma_1,\Sigma_2 \in \mathbb{R}^{n \times n}_{\geq 0}$, and a rational number
  $\epsilon \in (0,1)$.  \textbf{OUTPUT}: A number $z$ such that
$
    (1-\epsilon)D \leq z \leq (1+\epsilon)D,
$
  where
  $D \defeq \dtv\inp{\gaussian(\vct{\mu}_1, \Sigma_1), \gaussian(\vct{\mu}_2,
    \Sigma_2)}$.
\end{problem}
The case \(D = 0\) occurs only if \(\vec{\mu_1} = \vec{\mu_2}\) and
\(\Sigma_{1} = \Sigma_2\), which can be easily detected.  Thus, for ease of presentation,
we assume that the algorithm only has to consider inputs for which \(D > 0\).

\paragraph{Computational model} For ease of presentation, in the main body of
the paper, we present our results using the real-RAM model, where basic
arithmetic operations (comparison, addition, subtraction, multiplication,
division, logarithm, and square root) on real numbers take one unit of
time. This is one of the standard cost models in numerical analysis and
computational geometry (see, e.g. \citet[Section 1.4]{preparata-shamos-85} and
\citet[Section 5.1]{BC13}).  For the randomized version of our algorithm
(\cref{thm-gauss-rand}), we also assume that the algorithm can access randomness
in the following way: for any \(p \in [0, 1]\) that the algorithm constructs, it
can obtain, in one operation, a random bit sampled according to
\(\mathrm{Bernoulli}(p)\) and independent of all previous bits it has used.
However, our algorithm does \emph{not} require the ability to sample from
continuous distributions (e.g. Gaussian distributions).\footnote{Note, however,
  that in \Cref{sec:algor-with-appr}, we describe how to implement our
  algorithms given access only to an approximate diagonalization oracle; and the
  best known approximate diagonalization algorithms do require access to
  approximate samples from uniform or standard normal distributions.}  However,
as stated above, we do assume the ability to perform diagonalization of
symmetric matrices.  Our results assuming this computation model are stated
below in \cref{thm-gauss-det,thm-gauss-rand}.

\begin{theorem}\label{thm-gauss-det}
  There exists a deterministic algorithm that solves \multnormtv{} using at most
  $O\left(\frac{n^3}{\epsilon^2}\log^2\frac{n}{\epsilon D}\left(\log
        \frac{n}{\epsilon} + \log \log \frac{3}{D}\right) + \frac{n^2}{\epsilon}
      \log^3 \frac{n}{\epsilon D}\right)$ arithmetic operations along with
  diagonalizations of two square matrices of size at most $n \times n$.
\end{theorem}

\begin{theorem}\label{thm-gauss-rand}
  There exists a randomized algorithm that solves \multnormtv{} with probability
  at least $1 - \delta$ using at most
  $O\left(\frac{n^3}{\epsilon^3}\log \frac{n}{\epsilon D}\log\frac{1}{\delta} +
      \frac{n^2}{\epsilon} \log^3 \frac{n}{\epsilon D}\right)$ arithmetic
  operations along with diagonalizations of two square matrices of size at most
  $n \times n$.
\end{theorem}

\paragraph{The bit complexity model} Although the real-RAM model we use above is
convenient to work with (and also popular), a more realistic cost model is to
assume that the inputs are \emph{rational numbers} encoded as ratios of
integers specified in binary, and to account for costs using the total number of
\emph{bit} operations required to produce the desired output.  This is one of
the standard cost models in computational complexity theory and in
optimization~\citep{GLS93}.  Note that in this model, operations such as
logarithm, square root, and matrix diagonalization can only be carried out
approximately.\footnote{The floating point model~(see, e.g.,
  \cite{demmel_applied_1997}), perhaps the most widely used in practice, can be
  seen as a version of this model in which one fixes a precision \(k\) and a
  range parameter \(R\), and must approximate the result of every intermediate
  operation by a number of the form \(s \times 2^{m}\) where \(s, m\) are integers
  satisfying \(\abs{s} < 2^k \) and \(\abs{m} \leq R\).}  An algorithm in this
model is said to run in polynomial time if the total number of bit operations it
performs is at most a polynomial function of the total representation length of
the input.

In \cref{sec:prop-error-funct,sec:algor-with-appr}, we show how our algorithms continue to work efficiently even when
steps requiring exact computation of functions such as logarithm, square root
and matrix diagonalization are replaced by polynomial-time but approximate
versions of these functions implemented in the bit complexity model.  While we
focus on the real-RAM model in the main body of the paper, we include pointers
to the Appendix at points where such replacements may be required
when analyzing the algorithm in the bit complexity model.

\subsection{Discussion and Related Work} The problem of approximating
$\dtv(P, Q)$ to within additive error of \(\pm1/10\), given access only to
efficient samplers for $P$ and $Q$ is believed to be hard even for discrete
distributions on finite sets~\citep{SahaiV03}: solving it efficiently would
break certain cryptosystems~(see, e.g., \citet[Section 1.1]{BCHTV20}).  This
hardness result holds even if one is also given access to descriptions of the
circuits implementing the samplers.  If one is also in addition given access to
the probability density functions of $P$ and $Q$, then one can attempt a Monte
Carlo approach, but such an approach would require about $1/\epsilon^2$ samples for an
estimate with an \emph{additive} error $\epsilon$.  Such additive error algorithms for
the TV distance have also been studied by \cite{Kiefer18}, \cite{GMV20}, and
\cite{TXWSC24} in various settings. However, such additive error results can
become inefficient when transferred to the relative error setting because the TV
distance itself can be exponentially small in the underlying dimension (so that,
e.g. in the Monte Carlo approach, the required number of samples would also grow
exponentially). One is therefore forced to look for other methods in the
relative error setting, and this has been done in various discrete settings in
papers by \cite{BGMMPV23,GMM0V24} and \cite{FGJW23,FengLL24} (some of which were
already discussed above).

Prior work on approximating TV distance between Gaussians has focused on
achieving a \emph{fixed} relative error, but by means of closed form formulas;
see, e.g. \cite{DMR23} and references therein, \cite{BP24} for some recent
progress in a special case, and 
\cite{DMPR22} for similar results in the setting of mixtures of Gaussians.  The
quantitatively best such result we are aware of appears in the work of
\citet[Theorem 1.8]{AAL23}, who give a $100\sqrt{2}$-factor approximation for
the TV distance between multivariate Gaussians via a closed form formula. (In
contrast, note that \cref{thm-gauss-rand,thm-gauss-det} produce
$(1+\epsilon)$-factor multiplicative approximations for any $\epsilon > 0$.)
The case of one-dimensional Gaussians is related to the numerical evaluation of
the \emph{error function}, a problem of much interest in numerical
analysis~\citep{Chev12}.  In our analysis, we do use some of the above results
(see, e.g., \cref{prop:gau,lem:lower}).

\section{Preliminaries}

\paragraph{Total Variation Distance}
Let $P$ and $Q$ be two probability measures on a measurable space
$(\Omega,\mathcal{F})$. Their \emph{total variation distance} (TV distance) is
defined as $\dtv\inp{P,Q} \defeq \sup_{A \in \mathcal{F}}|P(A) - Q(A)|$.  We
often denote $\dtv\inp{P,Q}$ by $ \dtv\inp{X,Y}$, where $X \sim P$ and
$Y \sim Q$ are random variables with laws $P$ and $Q$ respectively.

For real vector valued
random variables with $\Omega = \mathbb{R}^n$, we assume
$\mathcal{F} = \mathcal{B}(\mathbb{R}^n)$ is the usual Borel $\sigma$-field.
In particular, if $P$ and $Q$ are two probability measures in $\mathbb{R}^d$
with probability density functions (PDFs) $p$ and $q$ respectively, then
$\dtv\inp{P,Q} = \frac{1}{2} \int_{x \in \mathbb{R}^d}|p(x) - q(x)| dx$.

We will use the following standard property of $\dtv$.
\begin{proposition}\label{prop-bijection}
Let $X$ and $Y$ be two random variables. For any measurable  $f$, it holds that
$
\dtv{ \inp{X,Y} } \geq \dtv{\inp{f(X),f(Y)}}.
$
If $f$ is further an invertible function, then 
$
 \dtv{ \inp{X,Y} } = \dtv{\inp{f(X),f(Y)}}.   
$
\end{proposition}

\paragraph{Gaussian Distribution}
A real valued random variable $X$ is said to be \emph{Gaussian} with \emph{mean}
$\mu$ and \emph{variance} $\sigma^2$ if it has the probability density function
(PDF) $f = f_{\mu,\sigma^2}$ (shortened to $f$ when $\mu$ and $\sigma^2$ are
clear from the context) given by
$f(x) = \frac{1}{\sqrt{2\pi \sigma^2}} \exp \left( -\frac{(x-\mu)^2}{2\sigma^2}
\right)$ for all $x \in \R$.  The probability distribution of $X$ is denoted by
$\gaussian(\mu,\sigma^2)$.  $\gaussian(0, 1)$ is called the \emph{standard}
Gaussian distribution.

A vector valued random variable $\vct{X}$ taking values in $\R^n$ is said to be
Gaussian with mean $\vct{\mu} \in \R^n$ and (symmetric, positive semi-definite)
\emph{covariance matrix} $\Sigma$ in $\R^n$ if for every vector $\vct{a} \in \R^n$,
the real valued random variable $\vct{a}^T\vct{X}$ has law
$\gaussian(\vct{a}^T\vct{\mu}, \vct{a}^T \Sigma \vct{a})$.  The distribution of $\vct{X}$ is denoted
$\gaussian(\vct{\mu}, \Sigma)$.  When $\vct{\mu} = \vct{0}$ and $\Sigma = I_n$ (the $n \times n$
identity matrix), the distribution is said to be the \emph{standard} Gaussian in
$n$-dimensions.  In general, when $\Sigma$ is diagonal, the individual
components $X_i$ of $\vct{X}$ are independent (real) Gaussian random variables with
variance $\Sigma_{i,i}$.  We collect below some standard facts about the
Gaussian distribution.

\begin{fact}\label{fact-normal}
  The random variable ${\vct{X}} \sim \gaussian({\vct{\mu}}, \Sigma)$ is identically
  distributed as $C{\vct{Z}} + {\vct{\mu}}$, where $\vct{Z} \sim \gaussian(\vct{0}, I_r)$ is
  a standard $r$-dimensional Gaussian, $C = \mathbb{R}^{n \times r}$ is a
  matrix such that $CC^T = \Sigma$, and $r \le n$ is the rank of the matrix
  $\Sigma$.  Thus, $\vct{X}$ is supported on the affine space
  $\Range(\Sigma) + {\vct{\mu}}$, where $\Range(\Sigma)$ is the column space of
  $\Sigma$.  Further, if ${\vct{Y}} = A{\vct{X}} + {\vct{b}}$, where $A \in \mathbb{R}^{m \times n}$
  and ${\vct{b}} \in \mathbb{R}^{m}$, then $\vct{Y}$ has the law
  $\gaussian(A{\vct{\mu}}+{\vct{b}}, A\Sigma A^T)$.
\end{fact}

The distribution $\gaussian({\vct{\mu}},\Sigma)$ is said to be
\emph{non-degenerate} if $\Sigma$ is invertible, and hence positive
definite. The probability density function (PDF)
$f = f_{{\vct{\mu}},\Sigma}: \R^n \rightarrow \R_{\geq 0}$ (shortened to $f$ when
${\vct{\mu}}$ and $\Sigma$ are clear from the context) of a non-degenerate
$\gaussian({\vct{\mu}},\Sigma)$ is given by
$
    f({\vct{x}}) \defeq  \frac{\exp\left( -\frac{1}{2} (\vct{x}-\vct{\mu})^T\Sigma^{-1} (\vct{x} - {\mu}) \right)}{\sqrt{(2\pi)^n \det(\Sigma)}}.
$

\paragraph{The Error Function}
The error function, \erf{\cdot}, defined as
$  \erf{x} \defeq \frac{2}{\sqrt{\pi}}\int\limits_0^x\exp(-t^2)\, dt,
$
for $x \in [-\infty, \infty]$, is a useful primitive arising in several numerical
applications.  Note that if $F$ denotes the cumulative distribution function
\emph{CDF} of $\gaussian(0,1)$, then for all $x \in [-\infty, \infty]$,
$  F(x) = \frac{1}{2}\inp{1 + \erf{\frac{x}{\sqrt{2}}}}.
$
The computation of the error function has been extensively studied.  We will
need the following result about the additive approximation of the error
function, which follows from methods reported by \cite{Chev12} (we include the
proof, along with an analysis in the bit complexity model, in
\cref{sec:prop-error-funct}). 

\begin{proposition}\label{prop:gau}
  There exists an algorithm such that given $\mu,\sigma^2$,
  $-\infty \leq a \leq b \leq \infty$ and $0 < \epsilon \leq 1/2$, it returns a
  real number $y \geq 0$ such that $|y - \int_{a}^b f(t)dt| \leq \epsilon$,
  where $f(t) = \frac{1}{\sqrt{2\pi \sigma^2}}e^{-\frac{(t-\mu)^2}{2\sigma^2}}$.
  The algorithm uses $O(\log^2(1/\epsilon))$ arithmetic operations.  \end{proposition}

\paragraph{Discrete Product Distributions}
We collect here known results on approximating TV-distances between products of
discrete distributions, that were discussed in the introduction.
\begin{problem}[\textbf{\disprodtv}]\label{prob-dis-prod-tv}
  \textbf{INPUT}: $2n$ discrete distributions  $P_i,Q_i$ for $1 \leq i \leq n$ over a finite domain $[M]=\{1,2,\ldots,M\}$, and a rational number $\epsilon \in (0,1)$.
  \textbf{OUTPUT}: A number $z$ such that
$
    (1-\epsilon)\dtv\inp{P,Q} \leq z \leq (1+\epsilon)\dtv\inp{P,Q},
$
where $P = P_1 \times \ldots \times P_n$ and $Q = Q_1 \times \ldots \times Q_n$ are two product distributions over $[M]^n$.
\end{problem} 
As mentioned in the introduction, the case $\epsilon = 0$ of the above problem
is \#P-complete~\citep{BGMMPV23}. The following two approximation algorithms are
known.

\begin{theorem}[\text{\cite{FengLL24}}]\label{lem:alg1}
  There exists a deterministic algorithm that solves \disprodtv{} in time $O(\frac{M n^2}{\epsilon}\log M \log \frac{n}{\epsilon \dtv \inp{P,Q}})$. 
\end{theorem}
\begin{theorem}[\text{\cite{FGJW23}}]\label{lem:alg2}
  There exists a randomized algorithm that solves \disprodtv{} in time
  $O(\frac{M n^2}{\epsilon^2}\log \frac{1}{\delta})$ with probability at least
  $1 - \delta$.
\end{theorem}
The running time in the above theorem is in terms of the number of arithmetic
operations.  Further, the randomized algorithm underlying \cref{lem:alg2} is assumed to have
the following access to randomness: for any \(p \in [0, 1]\) that the
algorithm constructs, it can obtain, in one operation, a random bit sampled
according to \(\mathrm{Bernoulli}(p)\) and independent of all previous bits it has
used. 

We will use both these algorithms as
subroutines in order to derive \cref{thm-gauss-det,thm-gauss-rand}.  However, in
\cref{sec:tv-dist-estim}, we give a new (and arguably simpler) proof of
correctness of the algorithm underlying \cref{lem:alg1}, as a corollary of our
framework for analyzing TV distance between product distributions.

\paragraph{Remarks on bit complexity} The proofs in \citet{FGJW23,FengLL24} show
that the algorithms in \cref{lem:alg1,lem:alg2} are also polynomial-time in
terms of bit complexity.  (In particular, note that when the inputs
\((P_i, Q_i)_{i=1}^{n}\) to \(\disprodtv{}\) are rational numbers specified as
ratios of integers denoted in binary, \(\dtv(P, Q)\) has a representation length
that is linear in the total input size.)  Further, the model of access to
randomness for the algorithm in \cref{lem:alg2} can be relaxed so that the
algorithm can only access a sequence of i.i.d.\ \(\mathrm{Bernoulli}(1/2)\)
random bits: this is essentially because for any \(p\) of the form
\(k/2^{\ell}\), where \(\ell\) is a positive integer and \(k \leq 2^l\) is an odd
integer, a sample from \(\mathrm{Bernoulli}(k/2^{\ell})\) can be obtained using
\(\ell\) i.i.d.\ bits sampled from \(\mathrm{Bernoulli}(1/2)\) (in fact, in
expectation, one only needs to sample two bits, even if \(p\) is not of the form
\(k/2^{\ell}\)).

Note, however, that in the algorithm of \cref{lem:alg1}, not the just
the number of \emph{bit} operations, but even the number of \emph{arithmetic
  operations} depends upon both the \emph{number} of input parameters and their total
representation length.  In other words, while this algorithm is polynomial-time
in term of bit complexity, it is not ``strongly
polynomial''~\cite[p.~32]{GLS93}.  In the context of \cref{lem:alg1}, a strongly
polynomial algorithm would have to be polynomial-time in terms of bit
complexity, and in addition, the number of arithmetic operations it performs
would be bounded above by a polynomial \emph{only} of \(n, M\) and
\(1/\epsilon\). To the best of our knowledge, it is not known whether a strongly
polynomial deterministic algorithm for \(\disprodtv{}\) exists.

\section{The Extension Distance }
\label{sec:tv-dist-estim}
In this section, we present a proof of correctness of the algorithm of
\cite{FengLL24} for discrete distributions by recasting its analysis in the
general framework of Radon-Nikodym derivatives.  For our application to
Gaussians, we do not need the results at this level of generality, but we
believe that presenting the ideas of \cite{FengLL24} in this general framework
both simplifies the presentation and can also be helpful for future extensions
of their algorithmic ideas.  Further, even in the case of discrete
distributions, our proof of correctness of their algorithm is different (and
arguably simpler) because it dispenses with one of the main innovations of
\cite{FengLL24} (which they referred to as the $\Delta_{MTV}$ distance) and
shows instead that their algorithm for product distributions can be analyzed in terms of a more direct
quantity which we call the \emph{extension distance} (see
\cref{def-ext-distance}).

We begin with some preliminary notions.  Our notations and conventions for the
standard notions described below follows the standard references \citet{Wil91}
and \citet[Appendix A.4]{Dur19}.  In particular, following these references, we
will agree that two random variables are said to be the same (with respect to a
measure) if they differ only on a set of measure $0$.

\begin{definition}[\textbf{Radon-Nikodym derivative~\cite[Theorems A.4.7 and A.4.8]{Dur19}}] Let $P$ and
$Q$ be probability measures on a measurable space $(\Omega, \F)$.  There exists a unique decomposition of $P$ into two finite non-negative measures
$P_1$ and $P_2$ such that $P = P_Q + P_Q^{\perp}$, $P_Q$ is \emph{absolutely
  continuous} with respect to $Q$ (denoted $P \ll Q$, and meaning that for every
event $A \in \F$, $Q(A) = 0$ implies $P_Q(A) = 0$), while $P_{Q}^{\perp}$ and $Q$
are mutually singular (meaning that there exists an event $A \in \F$ such that
$P_Q^{\perp}(A) = 0$ and $Q(A^c) = 0$).  Further, there exists a unique (with respect to
the probability measure $Q$) $\F$-measurable non-negative random variable $R$,
denoted $\diff{P_Q}{Q}$, such that for any event $A \in \F$,
 $   P(A) = P_Q^{\perp}(A) + \Ex_{Q}[RI_A],$
  where $I_A$ denotes the indicator function of $A$.  In case $P \ll Q$, we have
  $P_Q^{\perp} \equiv 0$, so that the first term on the right hand side of the
  above equation is zero, and $R = \diff{P}{Q}$ is referred to as the
  \emph{Radon-Nikodym derivative} of $P$ with respect to $Q$. Also,
  $\Ex_{Q}[R] \leq 1$, with equality holding if and only if $P \ll
  Q$.  \label{def:radon-niko}
\end{definition}

\begin{remark}\label{rem:ratio-remark}
  The random variable $R = \diff{P_Q}{Q}$, together with the law induced on it
  by the probability measure $Q$, was called the \emph{ratio} in the work of
  \cite{FengLL24}, and was denoted in that paper by the notation
  $\ratio{P}{Q}$. We also adopt this notation.  Note that the ratio is the same
  as the usual Radon-Nikodym derivative when $P \ll Q$.  The following
  definition (a rephrasing of similar definitions given by \cite{FengLL24})
  will be useful.
\end{remark}

\begin{definition}[\textbf{Valid ratio and independent products}]
  A non-negative random variable $R$ defined on a probability space
  $(\Omega, \F, Q)$ is said to be a \emph{valid ratio} if $\Ex[R] \leq 1$.
  Given valid ratios $R_i$ on $(\Omega_i, \F_i, Q_i)$ for $i \in {1, 2}$, the
  \emph{independent product} $R_1 \indep R_2$ of $R_1$ and $R_2$ is the valid
  ratio on $(\Omega_1\times \Omega_2, \F_1\times \F_2, Q_1\times Q_2)$ defined
  by $(R_1\indep R_2)(\omega_1, \omega_2) \defeq R_1(\omega_1)R_2(\omega_2)$.
  Note that the binary operation $\indep$ is associative, and that the random
  variables $R_1 \indep R_2$ and $R_2 \indep R_1$ have the same cumulative
  distribution function. \label{def:indep}
\end{definition}

We record the following standard observation about the Radon-Nikodym derivatives
in the setting of product measures.
\begin{fact}\label{fct-product}
  Let $(\Omega_1, \F_1)$ and $(\Omega, \F_2)$ be measurable spaces, and consider
  the product measurable space $(\Omega_1 \times \Omega_2, \F)$ where
  $\F = \F_1 \times \F_2$ denotes the product $\sigma$-algebra of $\F_1$ and
  $\F_2$.  Let $\pi_1$ and $\pi_2$ be probability measures on $(\Omega_1, \F_1)$
  and, similarly, let $Q$ and $P$ be probability measures on
  $(\Omega_{2}, \F_2)$.  Then for the product measures $\pi_1 \times Q$ and
  $\pi_2 \times P$ on $(\Omega, \F)$ we have
  $\inp{\ratio{\inp{\pi_2\times P}}{\inp{\pi_1 \times Q}}}(\omega_1, \omega_2) =
  \inp{\ratio{\pi_2}{\pi_1}}(\omega_1)\cdot \inp{\ratio{P}{Q}}(\omega_2)$.
  In terms of \cref{def:indep}, this can be written as
  ${\ratio{\inp{\pi_2\times P}}{\inp{\pi_1 \times Q}}} =
  \inp{\ratio{\pi_2}{\pi_1}} \indep \inp{\ratio{P}{Q}}$.
\end{fact}

It is well known that $\dtv(P, Q) = \Ex[(1-R)_+]$ when $R = \ratio{P}{Q}$
(where $|x|_+ \defeq \max\inb{x, 0}$; see \Cref{sec:proofs-omitted-from} for a
proof). This expression suggests defining the following functional.
\begin{definition}[\textbf{The TV functional}]\label{def-tv-functional}
  Let $R$ be a valid ratio on the probability space $(\Omega, \F, Q)$.  Then, we
  define
  $ TV(R) \defeq \Ex[(1-R)_+] = \Ex[(R-1)_+] + 1 - \Ex[R] \\
  = \frac{1}{2}\inp{\Ex[\abs{1 - R}] + 1 - \Ex[R]}, $
  where the equality of the three expressions follows by simple algebra. Thus,
  if $P$ and $Q$ are probability measures on $(\Omega, \F)$ such that
  $R = \ratio{P}{Q}$, then
 $   \dtv(P, Q) = TV(R).$
\end{definition}
An important property of the $TV$ functional is that $TV(R)$ depends only upon
the cumulative distribution function (\cdf) of the random variable $R$, and it
does not otherwise depend on the probability space on which $R$ has been
defined.  In particular, if $R$ and $S$ are valid ratios defined on possibly
different probability spaces, but if $R$ and $S$ have the same \cdf, then
$TV(R) = TV(S)$.  The following monotonicity properties of the $TV$ functional
are useful.

\begin{lemma}
  Let $R$ be a valid ratio defined on a probability space $(\Omega, \F, P)$.  Then,
  \begin{enumerate*}[label=(\arabic*), ref=\arabic*]
  \item \label{item:cR} For every $c \in [0, 1]$, the valid ratio $cR$ satisfies
    $TV(R) \leq TV(cR).$
  \item \label{item:conditional} If $\mathcal{G} \subset \F$ is a
    sub-$\sigma$-algebra, then the valid ratio
    $R_{\mathcal{G}} \defeq \Ex[R|\mathcal{G}]$ satisfies
    $TV(R_{\mathcal{G}}) \leq TV(R)$.  Further, if the event $\inb{R < 1}$ is
    $\mathcal{G}$-measurable, then the events $\inb{R < 1}$ and
    $\inb{R_{\mathcal{G}} < 1}$ are almost surely equivalent (i.e., the
    probability that exactly one of them happens is $0$), and
    $TV(R_{\mathcal{G}}) = TV(R)$.
  \item \label{item:product} If $S$ is a valid ratio defined on some probability
    space $(\Omega', \F', P')$ then $TV(R) \leq TV(R \indep S)$.
\end{enumerate*}\label{lem:tv-mon}
\end{lemma}

At the heart of the algorithmic framework of \cite{FengLL24} is a discretization
scheme for ratios.  Consider the following partition of the non-negative reals.

\begin{definition}[\textbf{$(\gamma, \delta)$-partition}] \label{def:partition}
  Fix $\gamma , \delta \in (0, 1)$, and let
  $m = 1 + \ceil{\ln(1/\gamma)/\ln(1+\delta)}$.  Define $a_0 = 1$, $a_{m} = 0$
  and $a _{k} \defeq 1 - (1+\delta)^{k-1}\gamma$ for integers $1 \leq k < m$, and
  the intervals $I_0 \defeq [1, 1] = \inb{1}$ and $I_k \defeq [a_k, a_{k-1})$
  for $1 \leq k \leq m$, which partition $[0, 1]$ into disjoint intervals.
  Similarly, define $J_k \defeq (1/a_{k-1}, 1/a_k]$ for $1 \leq k < m$ and
  $J_m \defeq (1/a_{m-1}, \infty)$.  The collection of disjoint intervals
  \begin{equation*}
    \cI_{{\gamma, \delta}} \defeq \inb{I_0} \cup \inb{I_i \st 1 \leq k \leq m}
    \cup \inb{J_i \st 1 \leq k \leq m}
  \end{equation*}
  is called a \emph{$(\gamma,\delta)$-partition} of $\Rpos$.  By a slight
  overloading of notation, we will also view $\cI_{\gamma,\delta}$ as a function
  from $\Rpos$ to the set $\cI_{{\gamma, \delta}}$ by setting
  $\cI_{\gamma, \delta}(x) = I$ for $x \in \Rpos$ when $I$ is the unique
  interval in $\cI_{\gamma,\delta}$ satisfying $x \in I$. \end{definition}

We can now describe the discretization procedure of \cite{FengLL24} in our
setting.
\begin{definition}[\textbf{$(\gamma,\delta)$-discretization}]\label{def:discret}
  Let $R$ be a valid ratio on some probability space $(\Omega, \F, P)$.  Let
  $\cI_{\gamma,\delta}$ be the partition of $\Rpos$ in \cref{def:partition}.
  The \emph{$(\gamma,\delta)$-discretization} of $R$ is the
  $\sigma(\cI_{\gamma,\delta}(R))$-measurable random variable
  $\tR_{\gamma,\delta}$ defined by
  \begin{equation}
    \label{eq:8}
    \tR = \tR_{\gamma,\delta} \defeq \Ex[R | \cI_{\gamma,\delta}(R)].
  \end{equation}
  In particular, it follows from \cref{item:conditional} of \cref{lem:tv-mon}
  that $\tR$ is also a valid ratio satisfying $TV(\tR) = TV(R)$, and that,
  almost surely, the events $\inb{R < 1}$ and $\{\tilde{R} < 1\}$ are
  equivalent. In fact, a very similar argument as in the proof of that lemma
  (given in \Cref{sec:proofs-omitted-from}) shows that for each
  $I \in \cI_{\gamma,\delta}$, the events $\inb{R \in I}$ and
  $\{\tilde{R} \in I\}$ are almost surely equivalent.
 \end{definition}

Our point of departure from the work of \cite{FengLL24} is in how the ``error''
introduced by the above discretization is measured.  \cite{FengLL24} defined a
quantity which they call $\Delta_{MTV}$, and which they relate to Le Cam's
deficiency, and measure the discretization error in terms of this quantity.
However, this quantity is defined in terms of total variation distances between
distributions realizing the appropriate ratios.  We follow instead a more direct
approach, and show that the discretization error can be measured in terms of the
following, arguably simpler, definition.

\begin{definition}[\textbf{Extension distance}] \label{def-ext-distance} Let
  $R_1$ and $R_2$ be valid ratios defined on some probability space
  $(\Omega, \F, P)$.  The \emph{extension distance} $\ext(R_1, R_2)$ between
  $R_1$ and $R_2$ is defined as
  \begin{equation}
    \label{eq:10}
    \sup_S\abs{TV(S \indep R_1) - TV(S \indep R_{2})},
  \end{equation}
  where the supremum is over all valid ratios $S$.  Recall that since $TV(R)$
  depends only upon the \cdf{} of $R$, we can restrict $S$ to the set of
  non-negative Borel measurable random variables with expectation at most $1$.
  From the definition itself, it is apparent that $\ext(\cdot, \cdot)$ is at
  least a \emph{pseudo-metric}: it is non-negative, symmetric, and satisfies the
  triangle inequality.  These properties will be sufficient for our purposes,
  and so we do not consider the question of whether it is a metric as well.
\end{definition}

The utility of the extension distance comes from the following result, which
relates it to the discretization procedure of \cite{FengLL24}.  Some of the
computations in the (short) proof of this result (given in \Cref{sec:proofs-omitted-from}) are
similar in spirit to parts of the analysis by \cite{FengLL24}, though there they
were performed with the goal of controlling their $\Delta_{MTV}$ distance.

\begin{theorem}\label{thm-ext-estimate}
  Let $\tR$ be the $(\gamma,\delta)$-discretization of a valid ratio $R$ defined
  on some probability space $(\Omega, \F, P)$.  Then
  $\ext(R, \tR) \leq \gamma + \delta TV(R)$.
\end{theorem}

As an application of the extension distance framework, we now analyze the
correctness of the algorithm of \cite{FengLL24} for computing the total
variation distance between products of discrete distributions.  We restate their
algorithm for reference as \cref{alg-dis-prod-tv}.

\begin{theorem}\label{thm-alg-dis-prod-correct}
  The algorithm \disprodtv{} is correct.
\end{theorem}

\begin{proof}
  We follow the notation used in \cref{alg-dis-prod-tv}.  For
  $1 \leq i \leq {n-1}$ define
  $S_i \defeq R_{i+1} \indep R_{i+2} \indep \dots \indep R_n$, and set
  $S_n \equiv 1$ (i.e., the trivial ratio that takes value $1 $ with probability
  $1$).  Note that $\Delta \leq d_{TV}(P, Q) = TV(Y_1\indep S_1)$.  Further, for
  each $1 \leq i \leq n - 1$, $\tilde{Y_{i}}\indep S_i = Y_{i+1} \indep S_{i+1}$
  (from \cref{line-alg-iplus1} of \cref{alg-dis-prod-tv}).  The algorithm's
  output is $Z = TV(\tilde{Y_n}) = TV(\tilde{Y_n} \indep S_n)$. Thus, we get
  \begin{multline}
    \label{eq:20}
    \abs{d_{TV}(P, Q) - Z} = \abs{TV(Y_1\indep S_1) - TV(\tilde{Y_{n}}\indep S_n)}\\
    = \abs{\sum_{i=1}^n\inp{TV(Y_i\indep S_i) - TV(\tilde{Y_{i}}\indep S_i)}}.
  \end{multline}
  We can now use the definition of the extension distance, followed by
  \cref{thm-ext-estimate} to estimate the $i$th term above as follows:
  $\abs{TV(Y_i\indep S_i) - TV(\tilde{Y_{i}}\indep S_i)} \leq \ext(Y_i,
  \tilde{Y_i}) \leq \gamma + \delta TV(Y_i).$
  Plugging in the values of $\gamma$ and $\delta$, we thus obtain
  \begin{equation}
    \label{eq:22}
    \abs{d_{TV}(P, Q) - Z} \leq \frac{\epsilon\Delta}{2}
    + \frac{\epsilon}{2n}\sum\limits_{1 \leq i \leq n}TV(Y_i).
  \end{equation}
  A direct inductive argument then shows that each $Y_i$ is a conditional
  expectation of $R_1 \indep R_2 \indep \dots \indep R_i$.  Applying
  \cref{item:conditional,item:product} of \cref{lem:tv-mon}, we thus obtain, for
  each $1 \leq i \leq n$:
  \begin{equation*}
    \begin{aligned}[t]     \label{eq:23}
      TV(Y_i) 
      &\leq TV(R_1 \indep \dots \indep R_i)\\
      & \leq TV(R_1
        \indep \dots \indep R_i \indep S_i)\\
      & = TV(R_1 \indep \dots \indep R_i \indep R_{i+1} \indep \dots \indep R_n)
      \\
      &= d_{TV}(P, Q).
    \end{aligned}
  \end{equation*}
  Substituting this in \cref{eq:22} and recalling that
  $\Delta \leq d_{TV(P, Q)}$, we obtain the claimed result.
\end{proof}

\begin{algorithm}
\caption{\disprodtv~\citep{FengLL24}\label{alg-dis-prod-tv}}
\KwIn{$n$ pairs of probability distributions $(P_i, Q_i)_{i=1}^n$ on the finite
  set $[M]$, and a rational number $\epsilon \in (0, 1)$.} \KwOut{A number $Z$ satisfying
  $(1-\epsilon)d_{TV}(P, Q) \leq Z \leq (1+\epsilon)d_{TV}(P, Q)$, where
  $P \defeq P_1\times P_{2} \times \dots \times P_n$ and
  $Q \defeq Q_1 \times Q_2 \times \dots \times Q_n$.}

Set $\Delta \leftarrow \max\limits_{1 \leq i \leq n} d_{TV}(P_i, Q_i)$\; 

Initialize $\gamma \leftarrow \frac{\epsilon\Delta}{2n} $ and
$\delta \leftarrow \frac{\epsilon}{2n}$\; 

Set $R_i \leftarrow \ratio{P_i}{Q_i}$ for $1 \leq i \leq n$, and
$Y_1 \leftarrow R_1$\;

\For{$i \leftarrow 1$ \KwTo $n$}{ Set $\tilde{Y_i}$ to be the $(\gamma, \delta)$-discretization of $Y_{i}$\;

  If $i < n$, set $Y_{i+1} \leftarrow \tilde{Y_i}\indep R_{i+1}$\;\label{line-alg-iplus1}
} 

\Return{$TV(\tilde{Y_n})$}\;

\end{algorithm}
Given its correctness, a runtime bound for \cref{alg-dis-prod-tv} can be
estimated easily, and is as given in \cref{lem:alg1}.  An argument similar to
(but simpler than) the proof of \cref{thm-alg-dis-prod-correct} also gives the
following.
\begin{lemma}\label{lem-prod-discrete}
  Fix $\gamma, \delta \in (0, 1)$. For $1 \leq i \leq n$, let $R_i$ be valid
  ratios defined on probability spaces $(\Omega_i, \F_i, P_i)$, and let $\tR_i$
  be their corresponding $(\gamma, \delta)$-discretizations.  Then
$ \big|TV(R_1 \indep R_2 \indep \dots \indep R_n) -TV(\tR_1 \indep \tR_1 \indep
    \dots \indep \tR_{n}) \big|
    \leq n\delta \kappa + n\gamma, $ where $\kappa \defeq \max_{1 \leq i \leq n}^nTV(R_i)$.
\end{lemma}

\section{TV-Distance between Gaussians}\label{sec:alg}
Consider \Cref{prob:TV}. Given two $n$-dimensional Gaussians
$\gaussian(\vct{\mu_1}, \Sigma_1), \gaussian(\vct{\mu_2}, \Sigma_2)$, we want to
approximate
$D = \dtv \inp {\gaussian({\vct{\mu_1}}, \Sigma_1), \gaussian({\vct{\mu_2}}, \Sigma_2)}$
within relative error $(1 \pm \epsilon)$.  Without loss of generality, we may
further assume that both $\Sigma_1$ and $\Sigma_2$ have the full rank
$n$. Otherwise, if $\Rank(\Sigma_1) \neq \Rank(\Sigma_2)$, then $D = 1$, since
one of the distributions then has zero probability mass on the support of the
other.  If $\Rank(\Sigma_1) = \Rank(\Sigma_2) = r < n$, then consider the
TV-distance between ${\vct{X}} \sim \gaussian(\vct{0}, \Sigma_1)$ and
${\vct{Y}} \sim \gaussian({\vct{\mu_2}} - \vct{\mu_1}, \Sigma_2)$, which is the same as $D$. By
\Cref{fact-normal}, ${\vct{X}}$ is supported on $\Range(\Sigma_1)$.  There are two subcases:
\begin{enumerate*}[label=(\arabic*)]
\item If $\Range(\Sigma_2) + (\vct{\mu_2} - {\vct{\mu}}) \neq \Range(\Sigma_1)$,
  then $D = 1$, as in the previous
  case.  \item If $\Range(\Sigma_2) + (\vct{\mu_2} - \vct{\mu_1}) = \Range(\Sigma_1)$,
  let $\Pi$ be an $r \times n$ matrix whose row vectors form an orthogonal basis
  of $\Range(\Sigma_1)$. Since $\Pi$ is a bijection between $\Range(\Sigma_1)$
  and $\mathbb{R}^{r}$, by \Cref{prop-bijection},
  $\dtv\inp{{\vct{X}},{\vct{Y}}} = \dtv \inp{ \Pi { \vct{X}}, \Pi { \vct{Y}}}$.
  However, we then have $\Pi \vct{X} \sim \gaussian(\vct{0}, \Pi \Sigma_1 \Pi^T )$ and
  $\Pi \vct{Y} \sim \gaussian(\Pi ( \vct{\mu_2} - \vct{\mu_1} ), \Pi \Sigma_2 \Pi^T )$, so that
  they are $r$-dimensional Gaussian vectors whose covariance matrices have full
  rank $r$.\end{enumerate*}
Using (unnormalized) Gram-Schmidt orthogonalization, the process described above
can be implemented using $O(n^3)$ arithmetic operations (and also in
polynomial-time in terms of bit complexity: see, e.g., \citet[Section
1.4]{GLS93}). From now on, we always assume that $\Sigma_1$ and $\Sigma_2$ have the full
rank $n$.

\paragraph{Reduction to Product Distributions}
We now use standard linear algebraic properties to reduce to the case of product
distributions.  This is the only part of our algorithm that needs access to a
diagonalization sub-routine.  (As discussed in the introduction,
\cref{sec:algor-with-appr} describes how to
implement a version of the next lemma in the bit complexity
model.) \begin{lemma}\label{lem:product}
  There exists an algorithm which given $(\vct{\mu_1},\Sigma_1)$ and
  $(\vct{\mu_2},\Sigma_2)$ (where $\Sigma_1, \Sigma_2$ are $n \times n$ positive
  definite matrices and $\vct{\mu_1}, \vct{\mu_{2}}$ are vectors in $\R^n$), outputs
  $(\vct{\mu},\Sigma)$, using $O(n^3)$ arithmetic operations and two
  diagonalizations of $n \times n$ symmetric matrices, such that $\Sigma$ is a
  diagonal matrix and
  $\dtv\inp{\gaussian(\vct{\mu_1}, \Sigma_1), \gaussian(\vct{\mu_2}, \Sigma_2)} =
  \dtv\inp{\gaussian(\vct{\mu}, \Sigma), \gaussian(\vct{0}, I_n)}$.
\end{lemma}
\begin{proof}
  Let $\vct{X} \sim \gaussian(\vct{\mu}_1, \Sigma_1)$ and
  $\vct{Y} \sim \gaussian(\vct{\mu}_2, \Sigma_2)$.  We use diagonalization to
  decompose $\Sigma_2$ as $Q_2\Lambda_2 Q_2^{T}$, where $Q_2$ is an orthogonal
  matrix and $\Lambda_2$ is a diagonal matrix. Define the matrix
  $A \defeq Q_2 \Lambda_2^{-1/2}Q_2^{T}$.  Now, we diagonalize the matrix
  $A \Sigma_1 A^T$ to obtain the decomposition $Q_1 \Lambda_1 Q_1^{T}$ (where
  $Q_1$ is orthogonal and $\Lambda_1$ is diagonal).  Define
  $\vct{\hat{X}} = Q_1^{T} A(\vct{X} - \vct{\mu_2})$ and
  $\vct{\hat{Y}} = Q_1^{T}A(\vct{Y} - \vct{\mu_2})$.  Note that the transformation
  $\vct{v} \mapsto Q_1^{T} A(\vct{v} - \vct{\mu_2})$ is invertible.  Thus, by
  \cref{prop-bijection},
  $\dtv\inp{\vct{X}, \vct{Y}} = \dtv\inp{\vct{\hat{X}}, \vct{\hat{Y}}}.$
  Finally, $\vct{\hat{X}} \sim \gaussian(\vct{\mu},\Sigma)$ and
  $\vct{\hat{Y}} \sim \gaussian(\vct{0}, I_n)$, where
  $\vct{\mu} = Q_1^{T} A(\vct{\mu_1} - \vct{\mu_2})$ and $\Sigma = \Lambda_1$.
  Apart from the two diagonalization operations, the algorithm only performs a
  constant number of matrix multiplications that cost a total of $O(n^3)$
  arithmetic operations.
\end{proof}

\paragraph{Discretizing Product Gaussian Distributions}
With \Cref{lem:product}, we only need to approximate the TV-distance between the
Gaussian $\gaussian(\vct{\mu},\Sigma)$ with diagonal covariance matrix $\Sigma$
and the standard Gaussian $\gaussian (0,I_n)$; both of which are continuous
product distributions.
We give the following reduction algorithm that transforms this task into the
task of estimating the TV-distance between two \emph{discrete} product
distributions.
\begin{lemma}\label{lem:reduction}
  There exists a deterministic algorithm, which given a vector
  $\vct{\mu} \in \R^n$, a positive $n \times n$ diagonal matrix $\Sigma$, and
  $\epsilon > 0$, outputs $2n$ discrete distributions
  $(\tilde{P}_i,\tilde{Q}_i)_{1 \leq i \leq n}$, using at most
  $O(\frac{n^2}{\epsilon} \log^3 \frac{n}{\epsilon \dtv \inp{\vct{X},\vct{Y}}})$ arithmetic
  operations, and satisfying the following properties (here
  $\vct{X} \sim \gaussian(\vec{\mu}, \Sigma)$ and
  $\vct{Y} \sim \gaussian(\vct{0}, I)$ are Gaussian random vectors):
  \begin{itemize}
    \item every $\tilde{P}_i$ and $\tilde{Q}_i$ is defined over the domain $[M]$ with $M = O(\frac{n}{\epsilon} \log \frac{n}{\epsilon \dtv \inp{X,Y}})$;
    \item the product distributions
      $\tilde{P} = \tilde{P}_1 \times \ldots \times \tilde{P}_n$ and
      $\tilde{Q} = \tilde{Q}_1 \times \ldots \times \tilde{Q}_n$
      satisfy $ \left\vert\dtv\inp{\tilde{P},\tilde{Q}} - \dtv \inp {X,Y} \right\vert \leq
      \frac{\epsilon}{3}\dtv \inp {\vct{X}, \vct{Y}}.  $ \end{itemize}
\end{lemma}

We now proceed to prove this lemma.  We start with the following estimate.  For
each $i \in [n]$, $X_i \sim \gaussian(\mu_i,\Sigma_{i,i})$ and
$Y_i \sim \gaussian(0,1)$, define
$    \Delta_i \defeq \frac{1}{200}\min\left\{1, \max\left\{ |\Sigma_{i,i}-1|,
        40|\mu_i|  \right\}\right\}$ and 
$    \Delta \defeq \max_{1 \leq i \leq n}\Delta_i.$
Then, one has the following estimate from previous works (see \Cref{sec:proofs-omitted-from-1} for a proof).
\begin{lemma}\label{lem:lower}
    $\Delta \leq \dtv \inp {\vct{X}, \vct{Y}} \leq 10^4 n  \Delta$.
\end{lemma}
Define $R_i$ to be the ratio
$\ratio{\gaussian(\mu_i, \Sigma_{ii})}{\gaussian(0, 1)}$ defined on the
probability space $(\R, \mathcal{B}, \gaussian(0, 1))$.  We set
\begin{equation}
  \label{eq:36}
  \gamma = \frac{\epsilon \Delta }{50n}, \qquad \delta = \frac{\epsilon}{50n},
\end{equation}
and let
$\tR_i$ denote the $(\gamma, \delta)$-discretization of $R_i$, for each
$1 \leq i \leq n$.  Now, let $P_i$ (respectively, $Q_i$) denote the probability
distribution on the finite set (see \cref{def:partition}) $\cI_{\gamma,\delta}$
defined by $P_i(J) = \Ex[R_i\cdot I[R_i \in J]]$ (respectively,
$Q_i(J) = \Pr{R_i \in J}$) for each interval $J \in \cI_{\gamma,\delta}$.  Since
$\tR_i = \Ex[R_i|\cI_{\gamma,\delta}(R_i)]$, it then follows that
$\tR_i = \ratio{P_i}{Q_i}$.

The exact computation of the discrete distribution $P_i$ and $Q_i$ involves
computing probabilities and expectations of Gaussian variables, which, in turn,
amounts to evaluation of the error function.  The following lemma shows that we
can compute an approximate version of these distributions.
\begin{lemma}\label{lem:one-dim}
  There exists an algorithm which, for each $i \in [n]$, computes a pair of
  discrete distributions $(\tilde{P}_i,\tilde{Q}_i)$ over $[M]$, where
  $M = O(\frac{n}{\epsilon}\log \frac{n}{\epsilon \Delta}) $, using
  $O(\frac{n}{\epsilon}\log^3 \frac{n}{\epsilon \Delta})$ arithmetic operations
  such that $\dtv (\tilde{P}_i, P_i) \leq \frac{\epsilon \Delta}{50 n}$ and
  $\dtv (\tilde{Q}_i, Q_i) \leq \frac{\epsilon \Delta}{50 n}$.
\end{lemma}
\begin{proof}
Fix an index $1 \leq i \leq n$. Let $f(x) = \frac{1}{\sqrt{2\pi \Sigma_{i,i}}}e^{-(x - \mu_i)^2/(2\Sigma_{i,i})}$ denote the density function of $\gaussian(\mu_i,\Sigma_{i,i})$ and $g(x)  = \frac{1}{\sqrt{2\pi}}e^{-x^2/2}$ denote the density function of $\gaussian(0,1)$.
Recall that
$\gamma = \frac{\epsilon \Delta }{50n}, \delta = \frac{\epsilon}{50n}$.

Let $M = |\cI_{\gamma,\delta}|$, so that, by definition of $\gamma, \delta$ and
$\cI_{\gamma,\delta}$,  
\[M =
O(\log(1/\gamma)/\log(1+\delta))=O\left(\frac{n}{\epsilon} \log \frac{n}{\epsilon
  \Delta}\right).\]  
For any $J \in \cI_{\gamma,\delta}$, we then have, by definition,
\begin{align}
  P_i(J) &= \Ex[R_i\cdot I[R_i \in J]] = \int\limits_{x \in D(J)}f(x) dx, \text{
           and}  \label{eq:24}\\
  Q_i(J) &= \Pr{R_i \in J} = \int\limits_{x \in D(J)}g(x) dx,\label{eq:25}
\end{align}
where the integration domain $D(J)$ is defined by
$D(J) \defeq \inb{R_i \in J} = \{z \mid \frac{f(z)}{g(z)} \in J \}$.
If $a \leq b$ are the endpoints of $J$ then the set\footnote{Depending on the
  definition of $J$, some $<$ signs should be replaced with $\leq$, but that
  does not affect the result of integration in this case.} $D(Z)$ is the set of
 all $z$ satisfying
$\ln\left(a\sqrt{\Sigma_{i,i}}\right) < -\frac{(z - \mu_i)^2}{2\Sigma_{i,i}} +
\frac{z^2}{2} < \ln\left(b \sqrt{\Sigma_{i,i}}\right)$. By solving the quadratic
inequality, we see that $D(J)$ is a union of at most two intervals.  From
\cref{eq:24,eq:25}, we thus obtain that the distributions $P_i$ and $Q_i$ can be
computed exactly if one has an exact oracle for computing the error function
$\erf{\cdot}$. However, since we can only compute such integrals approximately,
we compute two distributions $\tilde{P}_i$ and $\tilde{Q}_i$ over $[M]$ that
approximate $P_i$ and $Q_i$ respectively.

We show how to compute $\tilde{P}_i$. The distribution $\tilde{Q}_i$ can be
computed similarly.  Let
$\zeta = \frac{1}{10M} \cdot \frac{\epsilon \Delta}{50 n}$. For each
$J \in \cI_{\gamma,\delta}$, we need to compute $\int_{x \in D(J)}f(x)dx$ where
$D(J)$ is a union of at most two intervals.
We use \Cref{prop:gau} to compute the integration over each interval with
additive error $\zeta/2$.
By adding the results together, we obtain a number $p_J \geq 0$ such that
$|p_J - P_i(J)| \leq \zeta$.
One may consider defining $\tilde{P}_i$ as the distribution such that for any
$J \in \cI_{\gamma,\delta}$, $\tilde{P}_i(J) = p_J$. However, it may not hold
that $\sum_{J \in \cI_{\gamma,\delta}} p_J = 1$. Note, however, that
$|\sum_{J}p_J - 1| \leq M\zeta$.  We thus make the following adjustments on the
vector $(p_J)_{J \in \cI_{\gamma,\delta}}$ to define the distribution
$\tilde{P}_i$:
\begin{enumerate*}[label=(\arabic*), ref=\arabic*]
\item If $\sum_{J}p_J < 1$, we add $1 - \sum_J p_J$ probability mass to
  $\tilde{P}_i([1,1])$.
    \item If $\sum_{J}p_J > 1$, we remove $\sum_{J}p_J - 1$ probability mass in total for $\tilde{P}_i$ to make $\tilde{P}_i$ a valid distribution.
\end{enumerate*}
We can then bound the TV-distance as follows
$  \dtv \inp{\tilde{P}_i,P_i} = \frac{1}{2}\sum_{J \in \cI_{\gamma,\delta}}\left|\tilde{P}_i(J) - P_i(J)\right|
  \leq\, \frac{1}{2}\sum_{J \in \cI_{\gamma,\delta}}\left(|\tilde{P}_i(J) -p_J| + |p_J - P_i(J)|\right)
  \leq\, M\zeta < \frac{\epsilon \Delta}{50n}$.
By \cref{prop:gau}, the number of arithmetic operations needed for computing
each $p_J$ is $O(\log^2 (1/\zeta))$. The total number of arithmetic operations is
therefore
$O(M\log^2(1/\zeta)) = O(\frac{n}{\epsilon}\log^3\frac{n}{\epsilon \Delta})$.
\end{proof}

\begin{lemma}\label{lem:error}
  Let $\vct{X}, \vct{Y}, P_i, Q_i$ be as above.  For any distributions
  $(\tilde{P}_i,\tilde{Q}_i)_{1 \leq i \leq n}$ with
  $\dtv \inp {\tilde{P}_i, P_i} \leq \frac{\epsilon \Delta}{50n}$ and
  $\dtv \inp {\tilde{Q}_i, Q_i} \leq \frac{\epsilon \Delta}{50n}$, define the
  product distributions
  $\tilde{P} = \tilde{P}_1 \times \ldots \times \tilde{P}_n$ and
  $\tilde{Q} = \tilde{Q}_1 \times \ldots \times \tilde{Q}_n$. Then
\begin{align*}
\left \vert \dtv\inp{\tilde{P},\tilde{Q}} - \dtv \inp {\vct{X},\vct{Y}} \right\vert \leq \frac{\epsilon}{3}\dtv \inp {\vct{X},\vct{Y}}.
\end{align*}
\end{lemma}
\begin{proof}
  The following proof uses the extension distance based generalization of the
  framework of \cite{FengLL24} that was developed in
  \cref{sec:tv-dist-estim} above.
Define the product distributions $P = P_1 \times \dots \times P_n$ and
$Q = Q_{1} \times \dots \times Q_n$.  A simple coupling argument then shows that
$D_P \defeq \dtv \inp{P,\tilde{P}} \leq \frac{\epsilon \Delta}{50}$ and
$D_Q \defeq \dtv \inp {Q,\tilde{Q}} \leq \frac{\epsilon \Delta}{50}$. By triangle
inequality for $\dtv$, \begin{align}
  \label{eq:bd1}
  \abs{\dtv\inp{\tilde{P},\tilde{Q}} - 
    \dtv\inp{P, Q}} \leq D_P+ D_Q
  \leq \frac{\epsilon\Delta}{25}.
\end{align}
Next, we claim the following result, \begin{align}\label{eq:claim-1}
\left\vert\dtv \inp {P,Q} - \dtv \inp {X,Y}\right\vert\leq  \frac{\epsilon }{25}\dtv \inp{X,Y} .
\end{align}
To see this, recall that each $\tR_i = \ratio{P_i}{Q_i}$ is a
$(\gamma, \delta)$-discretization (with $\gamma, \delta$ as above) of
$R_{i} = \ratio{\mathcal{L}_{Y_i}}{\mathcal{L}_{X_i}}$ (where $\mathcal{L}_Z$
denotes the law of the random variable $Z$).  It thus follows from
\cref{fct-product} and \cref{def-tv-functional} that
$\dtv(P, Q) = TV(\tR_1 \indep \tR_2 \indep \dots \indep \tR_n)$ while
$\dtv(\vct{X}, \vct{Y}) = TV(R_1 \indep R_2 \indep \dots \indep R_n)$.
Eq.~\eqref{eq:claim-1} then is a direct consequence of \cref{lem-prod-discrete},
combined with the facts that $\Delta \leq \dtv\inp{\vct{X},\vct{Y}}$ and that
for each $i$, $TV(R_i) = \dtv(X_i, Y_i) \leq \dtv\inp{\vct{X},\vct{Y}}$.  The
lemma then follows by combining~\eqref{eq:bd1} and \eqref{eq:claim-1} with
$\Delta \leq \dtv(\vec{X}, \vec{Y})$.
\end{proof}

\begin{proof}[Proof of \Cref{lem:reduction}]
  We apply the algorithm of \cref{lem:one-dim} to each dimension to obtain the
  discretized distributions $\tilde{P}_i$ and $\tilde{Q}_i$. The running time
  can be verified by \Cref{lem:lower} and \Cref{lem:one-dim}, while the error
  bound is given by \Cref{lem:error}.
\end{proof}

\paragraph{The Final Algorithm}
By \Cref{lem:product} and
\Cref{lem:reduction}, to solve \multnormtv{}, we only need to approximate the TV-distance between the
discrete product distributions $\tilde{P},\tilde{Q}$ constructed in
\Cref{lem:reduction}.
Since $\tilde{P},\tilde{Q}$ are discrete, the resulting problem can now be
solved using existing algorithms in \Cref{lem:alg1} and \Cref{lem:alg2}.  Note
that our reduction algorithm is deterministic and uses
$O(n^3+\frac{n^2}{\epsilon} \log^3 \frac{n}{\epsilon D})$ operations (in
addition to diagonalizations of two $n \times n$ matrices), where $D$ is the
TV-distance between the two input Gaussian distributions.  This gives both
\cref{thm-gauss-det,thm-gauss-rand}.

\section{Conclusion} 
We give an efficient algorithm for approximating the TV-distance between two
multivariate Gaussians with relative error.  Along the way, we extend the
analysis framework of \cite{FengLL24} from discrete probability distributions to
general probability measures.  Several directions remain open; including TV
distance estimation for general log-concave distributions, graphical models, and
Gaussian-perturbed distributions; and approximations for other notions of
distance such as the Wasserstein distance.
A more open-ended question is to find as yet unexplored applications of
algorithms for estimating total variation distances, especially given the role
constant factor approximations for total variation distance between Gaussians
(discussed above) have played in various applications: see, e.g., \cite{AAL23}
for an example and \cite{DMPR22} for references.

\subsection*{Acknowledgments} PS acknowledges Rikhav Shah for helping with
technical questions about the results of \citet{shah_fast_2024}.  AB and PS
acknowledge the Mathematics of Data program at the Institute for Mathematical
Sciences, National University of Singapore, where this project was initiated.
PS acknowledgements support from the Department of Atomic Energy, Government of
India, under project no. RTI4001, from the Infosys-Chandrasekharan virtual
center for Random Geometry, and from SERB MATRICS grant number
MTR/2023/001547. The contents of this paper do not necessarily reflect the views
of the funding agencies listed above.

\renewcommand\bibname{References}

\appendix

\paragraph{Structure of the Appendix} \Cref{sec:prop-error-funct} provides the omitted proof of \cref{prop:gau} on
computational estimation of the error function, and also includes a sketch of
how to carry out the proof of \cref{lem:one-dim} in a weaker computational model
where logarithm and square root cannot be computed exactly.
\cref{sec:algor-with-appr} shows how to carry out the reduction to the product
case (described in \cref{sec:alg}) assuming access only to an approximate
diagonaliztion oracle that can be implemented using floating point arithmetic.
Finally, \Cref{sec:proofs-omitted-from} contains proofs omitted from
\cref{sec:tv-dist-estim}, while \cref{sec:proofs-omitted-from-1} contains proofs
omitted from \cref{sec:alg}.

\section{Properties of the Error Function}
\label{sec:prop-error-funct}
We start with the following proposition, which follows from ideas reported by
\cite{Chev12}.  Here, we denote by $\len{x}$ the \emph{representation length} in
bits of a rational $x$, when it is specified by its numerator and denominator.

\begin{proposition}\label{prop:err}
  There exists an algorithm such that given any small $0 < \epsilon \leq 1/2$
  and $x > 0$, it returns a real number $y \geq 0$ satisfying
  $|y-\erf{x}| \leq \epsilon$ using $O(\log^2(1/\epsilon))$ arithmetic
  operations. Further, for rational inputs, these operations may be assumed to
  be performed on rational numbers of representation length at most
  $\tilde{O}(\len{x} \cdot \log^2(1/\epsilon))$.
\end{proposition}

\begin{proof}[Proof of \cref{prop:err}]
  Without loss of generality, we assume $\epsilon \leq 10^{-4}$ (if the input
  $\epsilon$ is larger, we replace it with $10^{-4}$).  If
  $x > \ln(1/\epsilon)$, we simply return the number $y = \erf{+\infty} =
  1$. The error can be bounded as follows
\begin{align*}
    &|y - \erf{x}| \leq \frac{2}{\sqrt{\pi}}\int_{t > \ln\frac{1}{\epsilon}} e^{-t^2}dt\\
    \leq\,& \frac{2}{\sqrt{\pi}}\int_{t > \ln\frac{1}{\epsilon}} e^{-t\ln\frac{1}{\epsilon}}dt
    = \frac{2}{\sqrt{\pi}} \frac{1}{\ln\frac{1}{\epsilon}} \exp\left(-\ln^2\frac{1}{\epsilon}\right).
\end{align*}
Since $\epsilon$ is a small constant, the last term is at most $\epsilon$. 

Now, we assume $x < \ln (1/\epsilon)$. The following equation can be found,
e.g. in the work of~\citet[eq.~(1)]{Chev12}:
\begin{align*}
    \erf{x} = \frac{2}{\sqrt{\pi}}\sum_{n = 0}^{\infty}(-1)^n \frac{x^{2n+1}}{(2n+1)n!}.
\end{align*}
Let $N = 10\lceil \ln(1/\epsilon) \rceil^2$. We compute
$y = \frac{2}{\sqrt{\pi}}\sum_{n = 0}^{N}(-1)^n \frac{x^{2n+1}}{(2n+1)n!}$ using
$O(\log^2(1/\epsilon))$ arithmetic operations on rational numbers of
representation length at most $\tilde{O}(\len{x} \cdot \log^2(1/\epsilon))$
(at present, we ignore the small contribution to the error that comes from having to
use rational approximations of $\sqrt{\pi}$).  To bound the error, we note that
$n! \geq (n/e)^n$ so that
\begin{align*}
  |y - \erf{x}| &\leq \frac{2}{\sqrt{\pi}}\sum_{n > N} \frac{x^{2n+1}}{(2n+1) n!} < \frac{2}{\sqrt{\pi}}\sum_{n \geq N} \frac{x^{2n+1}e^n}{n^n} \\
  &< \frac{2x}{\sqrt{\pi}}\sum_{n > N}\left(\frac{x^2 e}{n}\right)^n.
\end{align*}
By the choice of $N$ and the fact $x < \ln(1/\epsilon)$, $\frac{x^2 e}{n} \leq \frac{x^2 e}{N} < \frac{1}{e}$. Hence,
\begin{align*}
  |y - \erf{x}| \leq \frac{2\ln(1/\epsilon)}{\sqrt{\pi}}\sum_{n > N}e^{-n}  < \epsilon/2,
\end{align*}
where the last inequality holds because $\epsilon$ is sufficiently small.
Finally, if $y < 0$, we let $y = 0$, which only decreases the error.  This
proves the first claim in the proposition.

It remains only to account for the error arising from having to use a rational
approximation to $1/\sqrt{\pi}$, when working with rational numbers only.
Propagating such an error through the above analysis shows that it is sufficient
for our purpose to compute $1/\pi$ up to an \emph{additive} error of
$\epsilon/10$.  This can be accomplished using $O(\log^2(1/\epsilon))$
arithmetic operations on rational numbers of representation length at most
$\tilde{O}(\log^2(1/\epsilon))$ using well-known series for the approximation of
$1/\pi$; see, e.g., eq.~(1) of \citet{chanDombsNumbersRamanujan2004} and the
work of \citet{ramanujanModularEquationsApproximations1914} and
\cite{borweinPiAGMStudy1987} for the original results.

\end{proof}

\begin{proof}[Proof of \cref{prop:gau}]
  First, let $a \gets (a - \mu)/\sqrt{2\sigma^2}$,
  $b \gets (b - \mu)/\sqrt{2\sigma^2}$ and
  $f(t) = \frac{1}{\sqrt{\pi}}e^{-t^2}$.  The value of the integration remains
  unchanged. If $ab \leq 0$, then
  $\int_{a}^b f(t)dt = \frac{1}{2}(\erf{|a|} + \erf{|b|})$. If $ab > 0$, say
  $a > 0$ and $b > 0$, then
  $\int_{a}^b f(t)dt = \frac{1}{2}(\erf{b} - \erf{a})$.\footnote{The case
    $a < 0$ and $b < 0$ follows by the symmetry of $\erf{\cdot}$.} Using
  \Cref{prop:err}, we compute the value of each error function with additive
  error $\epsilon$, which, by substitution in the above, gives us an additive
  error of at most $\epsilon$ for the integrals. Finally, if our answer $y < 0$,
  we let $y = 0$, which only decreases the error.  From \cref{prop:err}, the
  runtime is as claimed.
\end{proof}

\subsection{Proof of \texorpdfstring{\cref{lem:one-dim}}{Theorem Reference} in a Weaker Model of Computation}
The algorithm in \Cref{lem:one-dim} needs to solve the following quadratic inequality.
\[\ln\left(a\sqrt{\Sigma_{i,i}}\right) < -\frac{(z - \mu_i)^2}{2\Sigma_{i,i}} +
\frac{z^2}{2} < \ln\left(b \sqrt{\Sigma_{i,i}}\right).\]
The solution is a union of at most two intervals. One can solve for the endpoints of these intervals exactly if the computational model can exactly compute square roots and logarithms.

Consider the slightly weaker computation model that only provides approximate solutions to the above inequality.
Let $c$ and $d$ be two endpoints of one interval. Suppose that we can obtain  two approximate solutions $\tilde{c}$ and $\tilde{d}$ such that $|c - \tilde{c}| \leq \delta$ and  $|d - \tilde{d}| \leq \delta$ for some error bound $\delta$ in time $\mathrm{polylog}(\frac{1}{\delta})$. Let $f(x)  = \frac{1}{\sqrt{2\pi \Sigma_{i,i}}}e^{-(x - \mu_i)^2/(2\Sigma_{i,i})}$ be the density function in the proof of \Cref{lem:one-dim}.
The error can be controlled as 
\begin{align*}
  \left\vert \int_{\tilde{c}}^{\tilde{d}}f(x)dx - \int_{c}^d f(x)dx \right\vert \leq 2\delta \sup_xf(x) = 2\delta f(\mu_i) = \delta\sqrt{\frac{2}{\pi \Sigma_{i,i}}}.
\end{align*}
Similarly, for $g(x) =\frac{1}{\sqrt{2\pi}}e^{-x ^2/2}$, the error is at most $\delta \sqrt{\frac{2}{\pi}}$. 
One can take $\delta$ such that $\delta\sqrt{\frac{2}{\pi \min{(1,\Sigma_{i,i})}}}=O(\frac{\epsilon \Delta}{Mn})$.
Then the rest of the proof of \Cref{lem:one-dim} works exactly as before even in
this computational model.
The running time for solving the inequality in this model is
$\mathrm{polylog}(\frac{M n \max(1,\Sigma_{i,i}^{-1}) }{\epsilon \Delta})$.
The $\mathrm{polylog}(\max(1,\Sigma_{i,i}^{-1}))$ factor in the run time is
reasonable since one needs at least $\Theta(\log \Sigma_{i,i}^{-1})$ bits to
represent $\Sigma_{i,i} < 1$ in binary.

\section{Algorithm with Approximate Diagonalization Oracle}
\label{sec:algor-with-appr}
In the main paper (specifically, in \Cref{lem:product}) we assumed the existence
of an oracle that computes the \emph{exact} orthogonal diagonalization of a real
symmetric matrix.  It is common to work in such an exact real arithmetic model
for ease of presentation and in order to understand the high-level properties of
numerical algorithms. However, such a model of computation is not realistic
(see, e.g., \cite{DeyKRS23} for a discussion).  More realistic models include
floating point arithmetic, or a model of exact arithmetic with rationals
represented as ratios of integers.

In this section, we show that the algorithm in \cref{lem:product} remains
efficient even in the bit complexity model.  We do so by showing that this
algorithm can be implemented given access only to an appropriate
\emph{approximate} diagonalization oracle.

The oracle we work with is the following. Given a small enough positive $\delta$ and
a real symmetric positive definite matrix $A \in \mathbb{Q}^{n \times n}$ such that the input
size is $\len{A}$, we assume that the approximate diagonalization oracle
returns, using at most $\poly{n, \len{A}, \len{\delta}}$ bit operations, an
$n \times n$ diagonal matrix $\Lambda^{1/2}$ and an $n \times n$ matrix $Q$, all whose singular
values lie in the interval $[1-\delta/3, 1+\delta/3]$, such that
\begin{equation}
  \norm{A - Q \Lambda Q^T} \leq \delta \norm{A}.\label{eq-oracle-guarantee}
\end{equation}
The existence of such an oracle in the floating point model of computation (and
hence also in the bit complexity model) is known: see \citet[Theorem
4.2]{shah_fast_2024} for a state-of-the-art randomized algorithm implementing
such an oracle in the complex Hermitian case, and the earlier paper by
\citet[Remark 6.1]{BGKS23} for an explicit instantiation in the case of real
Hermitian matrices.\footnote{The oracles in the papers by \citet{shah_fast_2024}
  and \citet{BGKS23} return the diagonal matrix $\Lambda$ directly.  However, for ease
  of presentation, we include the entry-wise square root of the diagonal matrix
  $\Lambda$ into the oracle itself.  In the floating point model, this entry-wise
  square root only needs to be computed to an additive accuracy of
  $O(\delta\norm{A})$, and this can be done without significantly modifying the
  runtime characteristics of the oracle.} Note also that it follows from the
condition on $Q$ that $\norm{Q^{T}Q - I} \leq \delta$ (when $\delta$ is small enough).

We now show that there exists a polynomial-time algorithm that performs a
similar reduction as in \Cref{lem:product}, while only accessing the above
approximate diagonalization oracle.

For any real symmetric positive definite matrix $A$, define the condition number
\begin{align*}
    \kappa(A) = \frac{\lambda_{\max}(A)}{\lambda_{\min}(A)},
\end{align*}
where $\lambda_{\min}(A),\lambda_{\max}(A) > 0$ are the minimum and maximum
eigenvalues of $A$.

\begin{theorem}\label{thm:error}
  There exists an algorithm such that given $\epsilon > 0$,
  $(\vct{\mu_1},\Sigma_1)$, $(\vct{\mu_2},\Sigma_2)$ (where $\Sigma_1, \Sigma_2$ are
  $n \times n$ positive definite matrices in \(\mathbb{Q}^{n\times n}\) and
  $\vct{\mu_1}, \vct{\mu_{2}}$ are vectors in $\mathbb{Q}^n$),a lower bound $\Delta$ with
  $D = \dtv \inp {\gaussian(\vct{\mu}_1,\Sigma_1),\gaussian(\vct{\mu}_2,\Sigma_2)} \geq \Delta$, and
  an upper bound $\kappa>1$ with
  $\kappa(\Sigma_1),\kappa(\Sigma_2) \leq \kappa$, it returns
  $\vct{\mu} \in \mathbb{Q}^n$ and a diagonal positive definite
  $\Sigma \in \mathbb{Q}^{n \times n}$ such that
\begin{align*}
   \left \vert D - \dtv \inp{ \gaussian(\vct{\mu},\Sigma), \gaussian(0,I_n) } \right \vert \leq \frac{\epsilon}{100} D,
\end{align*}
The algorithm uses $O(n^3)$ arithmetic operations and two approximate
diagonalizations of $n \times n$ symmetric positive definite matrices of
representation length
\(\poly{\frac{n}{\epsilon}, \len{(\vec{\mu_1}, \vec{\mu}_2, \vec{\Sigma_1}, \vec{\Sigma_2})}}\) with
error bound $\delta =\poly{\frac{\epsilon\Delta}{n\kappa}}$.
\end{theorem}

Given the $\vct{\mu}$ and $\Sigma$ returned by the algorithm in
\Cref{thm:error}, one can use the discretization procedure in the main paper to
approximate the TV-distance between $\gaussian(\vct{\mu}, \Sigma)$ and
$\gaussian(0,I_n)$ within relative error $\epsilon/2$, which, by
\cref{thm:error}, approximates $D$ within relative error $\epsilon$.

The algorithm in \Cref{thm:error} needs to know a lower bound parameter $\Delta$ and
an upper bound parameter $\kappa$ in advance.  If, however,
$\vct{\mu}_1,\vct{\mu}_2$ and positive definite $\Sigma_1,\Sigma_2$ (assume
$\vct{\mu}_1 \neq \vct{\mu}_2$ or $\Sigma_1 \neq \Sigma_2$, otherwise $D = 0$) have rational entries
with numerators and denominators that are all integers represented by at most
$\poly{n}$ bits, then $\kappa = 2^{\poly{n}}$ (see~\cite{Sri23,DeyKRS23}) and
$\Delta = 2^{-\poly{n}}$ (see~\Cref{lem:TVlower} below). The overall running time is
thus
$\poly{\frac{n}{\epsilon}, \len{(\vec{\mu_1}, \vec{\mu}_2, \vec{\Sigma_1}, \vec{\Sigma_2})}, \len{\delta}}
= \poly{\frac{n}{\epsilon}, \len{(\vec{\mu_1}, \vec{\mu}_2, \vec{\Sigma_1}, \vec{\Sigma_2})}}$.

The following standard upper bound on TV-distance will be used in the proof of
\Cref{thm:error}. It follows from the exact expression for the KL divergence
between Gaussians, and Pinkser's inequality.
\begin{fact}[\textbf{see, e.g., \citet[Proposition 2.1]{DMR23}}]\label{fact:TV}
  Let $\vct{\mu}_1,\vct{\mu}_2 \in \mathbb{R}^n$ and
  $\Sigma_1,\Sigma_2 \in \mathbb{R}^{n \times n}$ be two covariance matrices
  with rank $n$.  The TV-distance between $\gaussian(\vct{\mu}_1,\Sigma_1)$ and
  $\gaussian(\vct{\mu}_2,\Sigma_2)$ is at most
\begin{align*}
  \frac{1}{2}\sqrt{\mathrm{tr}(\Sigma_1 \Sigma_2^{-1}) - n + (\vct{\mu}_1-\vct{\mu}_2)^T \Sigma_1^{-1}(\vct{\mu}_1-\vct{\mu}_2) + \ln\frac{\mathrm{det}(\Sigma_1)}{\mathrm{det}(\Sigma_2)} }.
\end{align*}
In particular, if $\vec{\mu}_1 = \vct{\mu}_2$ and $\Sigma_2 = I_n$, then the TV-distance is at most $ \frac{1}{2}\sqrt{\mathrm{tr}(\Sigma_1) - n  + \ln{\mathrm{det}(\Sigma_1)}}.$
\end{fact}

In the proof of \cref{thm:error}, we will repeatedly use the following
standard facts from linear algebra.
\begin{fact}[\textbf{Weyl's inequality, see, e.g., \citet[Theorem
    5.1]{demmel_applied_1997}}] If $A$ and $B$ are $n \times n$ real symmetric
  matrices with eigenvalues $\alpha_1 \geq \alpha_2 \geq \dots \geq \alpha_{n}$
  and $\beta_1 \geq \beta_2 \geq \dots \beta_n$, then for each
  $1 \leq i \leq n$, $\abs{\alpha_i - \beta_i} \leq \norm{A -
    B}$.  \label{item-lin-fact-weyl}
\end{fact}
\begin{fact}
  If $A$ is an $n \times n$ matrix then $\abs{\tr(A)} \leq n
  \norm{A}$. \label{item-lin-trace}
\end{fact}

In particular, we record the following simple consequence of these facts.
\begin{claim}\label{clm-estimate-lambda}
  Let $A$ be an $n \times n$ positive definite matrix. Let
  $\delta < 1/\kappa(A)$ be a small enough positive number, and suppose that
  $Q, \Lambda$ are obtained by calling the above approximate diagonalization
  oracle on $A$ with a small error parameter $\delta>0$.  Then,
  \begin{itemize}
    \item $\lambda_{\max}(\Lambda) \leq (1+\delta)^3\lambda_{\max}(A) =
      (1+\delta)^3\norm{A}$, and
    \item
      $\lambda_{\min}(\Lambda) \geq \frac{(1 -
        \delta\kappa(A))\lambda_{\min}(A)}{(1 + \delta)^2}.$
    \end{itemize}
    Thus,
    $\kappa(\Lambda) \leq (1+\delta)^5\frac{\kappa(A)}{1 - \delta\kappa(A)}$.
\end{claim}
\begin{proof}
  Using \Cref{item-lin-fact-weyl},
  $ \lambda_{\max}(Q\Lambda Q^T) - \lambda_{\max}(A) \leq \Vert A - Q\Lambda Q^T
  \Vert \leq \delta\Vert A \Vert = \delta \lambda_{\max}(A)$. Note that
  $(1-\delta/3)^2 \lambda_{\max}(\Lambda) \leq \lambda_{\max}(Q\Lambda
  Q^T)$. The first item thus holds when $\delta > 0$ is small.  For the second
  item, by \Cref{item-lin-fact-weyl},
  $\lambda_{\min}(A) - \lambda_{\min}(Q\Lambda Q^T) \leq \Vert A - Q\Lambda Q^T
  \Vert \leq \delta\Vert A \Vert = \delta \lambda_{\max}(A)$. Note that
  $(1+\delta/3)^2\lambda_{\min}(\Lambda) \geq \lambda_{\min}(Q\Lambda
  Q^T)$. Putting the two inequalities together proves the second item.
\end{proof}

\begin{proof}[Proof of \Cref{thm:error}]
  Let $X \sim \gaussian(\vct{\mu}_1,\Sigma_1)$ and
  $Y \sim \gaussian(\vct{\mu}_2,\Sigma_2)$.  We use the approximate
  diagonalization oracle on $\Sigma_2$ with error bound
  $\delta = \poly{\frac{\epsilon\Delta}{n\kappa}}$ chosen small enough, to
  compute $ \tilde{Q}_2$ and $\tilde{\Lambda}_2^{1/2}$ with the guarantees of
  \cref{eq-oracle-guarantee}.  Define $\tilde{A} \defeq \tilde{Q}_2\tilde{\Lambda}_2^{-1/2}\tilde{Q}_2^T$;
  note that $\tilde{A}$ is symmetric.  Consider the random variable
  $\bar{X} = \tilde{A}X-\tilde{A}\vct{\mu}_2$, which follows the law of
  $\gaussian(\tilde{A}(\vct{\mu}_1 - \vct{\mu}_2),\tilde{A} \Sigma_1
  \tilde{A}^T)$, and the random variable
  $\bar{Y}' = \tilde{A}Y-\tilde{A}\vct{\mu}_2$, which follows the law of
  $\gaussian(0,\tilde{A} \Sigma_2 \tilde{A}^T)$. Since $\tilde{A}$ is
  invertible, by \Cref{prop-bijection}, we have
\begin{align}\label{eq-t1}
    D = \dtv(X,Y) =  \dtv \inp{\bar{X},\bar{Y}'}.
\end{align}
However, to mimic the proof of \Cref{lem:product}, we replace $\bar{Y'}$ with
$\bar{Y} \sim \gaussian(0,I_n)$.  Note that $\bar{Y}'$ may not be a standard
Gaussian because we use only an approximate diagonalization oracle.  We thus
need to bound the total variation distance between $\bar{Y}$ and
$\bar{Y'}$. Using \Cref{fact:TV},
\begin{equation}
  \dtv\inp{\bar{Y},\bar{Y}'} \leq \frac{1}{2}\sqrt{\mathrm{tr}(\tilde{A} \Sigma_2 \tilde{A}^T) - n  + \ln{\mathrm{det}(\tilde{A} \Sigma_2 \tilde{A}^T)}}.   \label{eq:34}
\end{equation}

Using the invariance of trace under cyclic shifts followed by the guarantee of
the approximation diagonalization oracle along with \cref{item-lin-trace} above,
we get (we also use that $\delta$ is small enough so that
$(1 + \delta)^k \leq 1.1$ for $k \leq 10$).
\begin{align}
  &\tr(\tilde{A}\Sigma_2\tilde{A}^{T}) \\
  &= \tr(\tilde{A}^{T}\tilde{A}\Sigma_2)\\
& \leq \tr(\tilde{A}^{T}\tilde{A}\tilde{Q}_2\tilde{\Lambda}_2\tilde{Q}_2^{T}) + n\delta \norm{\Sigma_2} \norm{\tilde{A}}^2\label{eq-tri}\\
  &\leq \tr(\tilde{A}^{T}\tilde{A}\tilde{Q}_2\tilde{\Lambda}_2\tilde{Q}_2^{T}) + n\delta(1+\delta)^4\kappa(\Sigma_2)\frac{\lambda_{\min}(\Sigma_2)}{\lambda_{\min}(\tilde{\Lambda}_2)}\\
  &\leq \tr(\tilde{Q}_2^{T}\tilde{A}^{T}\tilde{A}\tilde{Q}_2\tilde{\Lambda}_2) 
    + \frac{2n\delta\kappa}{1 - \delta\kappa}.\label{eq:35}
\end{align}
Here, in the first inequality, we use $ \tr(\tilde{A}^{T}\tilde{A}\Sigma_2) \leq \tr(\tilde{A}^{T}\tilde{A}\tilde{Q}_2\tilde{\Lambda}_2\tilde{Q}_2^{T}) + |\tr(\tilde{A}^{T}\tilde{A}(\Sigma_2 -\tilde{Q}_2\tilde{\Lambda}_2\tilde{Q}_2^{T}) )|$, where we bound the second term by \Cref{item-lin-trace}.
In the last inequality, we use \cref{clm-estimate-lambda} along with the
approximate diagonalization oracle's guarantees.  Now, to estimate the first
term, we expand $\tilde{A} = \tilde{Q}_2\tilde{\Lambda}_2^{-1/2}\tilde{Q}_2^T$,
and note that $\norm{\tilde{Q}_2^T\tilde{Q}_2 - I} \leq \delta$.  Repeatedly
using this last estimate along with \cref{item-lin-trace} above then gives
\begin{align}
  \tr(\tilde{A}\Sigma_2\tilde{A}_2^{T})  &\leq \tr{I} +  \frac{7n\delta\lambda_{\max}(\tilde{\Lambda}_2)}{\lambda_{\min}(\tilde{\Lambda}_2)} + \frac{2n\delta\kappa}{1 - \delta\kappa}\nonumber \\ 
&\leq n + \frac{10n\delta\kappa}{1 - \delta\kappa},
  \label{eq:33}
\end{align}
where we again use \cref{clm-estimate-lambda} along with the approximate
diagonalization oracle's guarantee in \cref{eq-oracle-guarantee} to get the last
inequality.

Next, we need to bound the determinant of $\tilde{A} \Sigma_2 \tilde{A}^T$,
where $\tilde{A} = \tilde{Q}_2\tilde{\Lambda}_2^{-1/2}\tilde{Q}_2^T$. Let
$\lambda_i$ for $1 \leq i\leq n$ denote the eigenvalues of $\Sigma_2$ arranged
in descending order. Recall
that $\kappa(\Sigma_2) \leq \kappa$. We have
\begin{align}
\det(\tilde{A} \Sigma_2 \tilde{A}^T) &= \det(\tilde{A})^2\det(\Sigma_2)\nonumber \\
&= \det(\tilde{Q}_2)^6\frac{\det(\Sigma_2)}{\det(\tilde{Q}_2\tilde{\Lambda}_2\tilde{Q}_2^T)}\nonumber \\
&\overset{(\ast)}{\leq} (1+\delta/3)^{6n}\prod_{i=1}^n \frac{\lambda_i}{\lambda_i - \delta \lambda_1}\nonumber \\ 
&\leq \exp \left( \frac{3n\delta\kappa}{1-\delta \kappa}\right),
\label{eq:28}
\end{align}
where $(\ast)$ holds because of Weyl's inequality in \Cref{item-lin-fact-weyl} (with $A = \Sigma_2$ and $B =\tilde{Q}_2\tilde{\Lambda}_2\tilde{Q}_2^T $) and the guarantees from the
oracle.  As our $\delta = \poly{\frac{\epsilon\Delta}{n\kappa}}$ is chosen small
enough, we thus get (by substituting \cref{eq:28,eq:33} into \cref{eq:34})
\begin{align}\label{eq-t2}
  \dtv\inp{\bar{Y},\bar{Y}'} \leq \frac{1}{2}\sqrt{ \frac{13n \delta\kappa}{1-\delta \kappa}}  
  \leq \frac{\epsilon \Delta}{200} \leq \frac{\epsilon D}{200}.
\end{align}

The final step in the proof of \Cref{lem:product} is to transform $\bar{X}$ into
a product distribution. Let $\bar{\Sigma}_1=\tilde{A}\Sigma_1\tilde{A}^T$ be the
covariance matrix of $\bar{X}$.  We use the approximate oracle on
$\bar{\Sigma}_1$ with a small error bound
$\delta = \poly{\frac{\epsilon\Delta}{n\kappa}}$ to obtain $\tilde{Q}_1$ and
$\tilde{\Lambda}_1^{1/2}$ satisfying the guarantees of
\cref{eq-oracle-guarantee}: in particular, we have
$\norm{\bar{\Sigma}_1 - \tilde{Q}_1\tilde{\Lambda}_1\tilde{Q}_1^T} \leq \delta
\norm{\bar{\Sigma}_1}$.
We then let $\hat{X}' = \tilde{Q}_1^T \bar{X}$ and
$\hat{Y}' = \tilde{Q}_1^T \bar{Y}$. Then, as $\delta$ is small enough,
$\tilde{Q}_1$ is invertible so that
\begin{align}\label{eq-t3}
   \dtv(\hat{X}', \hat{Y}') = \dtv(\bar{X},\bar{Y}). 
\end{align}

Define first $\hat{Y}$ to have the same distribution as
$\bar{Y} \sim \gaussian(0, I)$.  Note that
$\hat{Y}' \sim \gaussian(0, \tilde{Q}_1^T\tilde{Q}_1)$.  The guarantees from the
approximate diagonalization oracle yield that $\norm{\tilde{Q}_1^T\tilde{Q}_1 - I } \leq \delta$.
Thus, \cref{fact:TV} yields that
\begin{align}
  \dtv(\hat{Y}, \hat{Y}') &\leq \frac{1}{2}\sqrt{\tr(\tilde{Q}_1^T\tilde{Q}_1) -
    n + \log\det(\tilde{Q}_1^T\tilde{Q}_1)}\nonumber \\
& \leq \sqrt{n\delta} \nonumber \\
& \leq \frac{\epsilon \Delta}{200} \leq \frac{\epsilon D}{200},
\label{eq:32}
  \end{align}
where the last two inequalities come from the small choice of $\delta$.

We now consider $\hat{X}'$.
The covariance matrix of $\hat{X}'$ is $\Sigma' = \tilde{Q}_1^T \bar{\Sigma}_1 \tilde{Q}_1$ and the mean vector of  $\hat{X}'$ is $\tilde{Q}_1^{T}\tilde{A}(\vct{\mu}_1 - \vct{\mu}_2)$.
Since $\tilde{Q}_1$ is obtained using the approximate diagonalization oracle, $\Sigma'$ may not be diagonal.
Define a product Gaussian $\hat{X}$ with the same mean vector $ \tilde{Q}_1^{T}\tilde{A}(\vct{\mu}_1 - \vct{\mu}_2)$ but a diagonal covariance matrix $ \tilde{\Lambda}_1$. By \Cref{fact:TV},
\begin{equation}
\label{eq:29}
\dtv (\hat{X}', \hat{X}) \leq  \frac{1}{2}\sqrt{\mathrm{tr}(\Sigma'\tilde{\Lambda}_1^{-1}) - n  + \ln\frac{\det(\Sigma')}{\det(\tilde{\Lambda}_1)}}.   
\end{equation}
In a manner analogous to \cref{eq-tri,eq:35,eq:33} above, we then get that
\begin{equation}
\label{eq:27}
\mathrm{tr}(\Sigma'\tilde{\Lambda}_1^{-1}) =
\tr(\tilde{Q}_1^{T}\bar{\Sigma}_1\tilde{Q}_1\tilde{\Lambda}^{-1}_1) \leq n +
\frac{10n\delta\kappa(\bar{\Sigma}_1)}{1 - \delta\kappa(\bar{\Sigma}_1)}, 
\end{equation}
provided $\delta < 1/\kappa(\bar{\Sigma}_1)$ is small enough.  We thus need to bound
the condition number of $\bar{\Sigma}_1=\tilde{A}\Sigma_1\tilde{A}^T$. 
Note that by the small enough choice of $\delta$, all of the matrices
$\Sigma_1,\bar{\Sigma}_1$, and $\tilde{A}$ are positive definite. We then have
\begin{align*}
    \lambda_{\max}(\bar{\Sigma}_1) &= \Vert \bar{\Sigma}_1 \Vert\\ &\leq \Vert
                                     \tilde{A} \Vert^2 \cdot \Vert \Sigma_1
                                     \Vert \\ &\leq
                                     (1+\delta)^4\lambda_{\min}^{-1}(\tilde{\Lambda}_2)
                                     \lambda_{\max}(\Sigma_1)\text{, and}\\
    \lambda_{\min}^{-1}(\bar{\Sigma}_1) &= \Vert \bar{\Sigma}_1 ^{-1} \Vert \\& \leq
                                          \Vert \tilde{A}^{-1} \Vert^{2} \cdot
                                          \Vert \Sigma_1^{-1} \Vert\\ &\leq
                                          (1+\delta)^4 \lambda_{\max}(\tilde{\Lambda}_2)\lambda_{\min}^{-1}(\Sigma_1).
\end{align*}
Hence,
$\kappa(\bar{\Sigma}_1) \leq (1+\delta)^8\kappa(\tilde{\Lambda}_2)
\kappa(\Sigma_1)$. By assumption, $\kappa(\Sigma_1) \leq \kappa$.  By
\cref{clm-estimate-lambda} applied with $A$ in the statement of the claim set to
$\Sigma_2$ and $\Lambda$ to $\tilde{\Lambda}_2$, we thus get that (for $\delta$
small enough) $\kappa(\bar{\Sigma}_1) \leq \frac{4 \kappa^2}{1-\delta\kappa}$
(since $\kappa(\Sigma_1)$ and $\kappa(\Sigma_2)$ are both at most $\kappa$).
Substituting this bound in \cref{eq:27}, we get (again for
$\delta = \poly{\frac{\epsilon\Delta}{n\kappa}}$ small enough)
\begin{equation}
\label{eq:30}
\tr(\Sigma'\tilde{\Lambda}_1^{-1}) \leq n +  200n\delta\kappa^2.
\end{equation}

Next, we need to bound the ratio of determinants, which is (by a calculation
similar to \cref{eq:28}, and using the estimate on $\kappa(\bar{\Sigma}_1)$
above)
\begin{align}
\frac{\det(\Sigma')}{\det(\tilde{\Lambda}_1)} 
&= \det(\tilde{Q}_1)^4\frac{\det(\bar{\Sigma}_1)}{\det(\tilde{Q}_1\tilde{\Lambda}_1\tilde{Q}_1^T)}\nonumber \\
&\leq \exp \left( \frac{3n\delta \kappa(\bar{\Sigma}_1)}{1 -
    \delta\kappa(\bar{\Sigma}_1)}\right) \nonumber\\ &\leq \exp(24n\delta\kappa^2).
\label{eq:31}
\end{align}
We now substitute the bounds in \cref{eq:30,eq:31} into the bound in
\cref{eq:29} for the total variation distance between $\hat{X}$ and $\hat{X}'$.
Since $\delta = \poly{\frac{\epsilon\Delta}{n\kappa}}$ is small enough, the
total variation distance between $\hat{X}$ and $\hat{X}'$ is thus bounded as
\begin{align}\label{eq-t4}
  \dtv\inp{\hat{X},\hat{X}'} \leq \frac{1}{2}\sqrt{224n\delta\kappa^2}
  \leq \frac{\epsilon \Delta}{200} \leq \frac{\epsilon D}{200}.  
\end{align}

Let $\vct{\mu} = \tilde{Q}_1^{T}\tilde{A}(\vct{\mu}_1 - \vct{\mu}_2)$ and
\(\Sigma = \tilde{\Lambda}_1\) be the mean vector and covariance matrix of the product
Gaussian $\hat{X}$, as computed above. Recall that
$\hat{Y} \sim \gaussian(0,I_n)$. By triangle inequality,
\begin{align*}
    &\dtv(\hat{X},\hat{Y})\\
    &\leq  \dtv\inp{\hat{X},\hat{X}'} +  \dtv\inp{\hat{X}',\hat{Y}'}
     + \dtv\inp{\hat{Y}',\hat{Y}}
\tag{\text{by triangle inequality}}\\
    &\leq  \frac{\epsilon D}{100} +  \dtv(\bar{X},\bar{Y})\tag{\text{by~\cref{eq-t3,eq-t4,eq:32}}}\\
    &\leq \frac{\epsilon D}{100} +  \dtv(\bar{X},\bar{Y}') + \dtv(\bar{Y},\bar{Y}') \tag{\text{by triangle inequality}}\\
    &\leq \frac{\epsilon D}{50} + D.\tag{\text{by~\cref{eq-t1,eq-t2}}}
\end{align*}
A similar triangle inequality calculation starting with
$D = \dtv(X, Y) = \dtv(\bar{X}, \bar{Y}')$ (see \cref{eq-t1}) shows that
\begin{align*}
  \dtv(\hat{X},\hat{Y}) \geq D - \frac{\epsilon D}{50},
\end{align*}
which completes the proof.  The running time of the algorithm is almost the same
as the running time of the algorithm in \Cref{lem:product} except that the two
algorithms use different oracles.
\end{proof}

Finally, we give a lower bound $\Delta$ if all inputs are given as rational
numbers.  Following the notation of \cite{DeyKRS23},
let $\mathbb{Z}\llangle a \rrangle$ denote all integers of bit length at most
$a$, and 
let $\mathbb{Q}\llangle a \rrangle$ denote all rational numbers $p/q$ with $p,q \in \mathbb{Z}\llangle a \rrangle$.
\begin{lemma}\label{lem:TVlower}
  Suppose that $\vct{\mu}_1,\vct{\mu}_2 \in \mathbb{Q}\llangle a\rrangle^n$ and
  $\Sigma_1,\Sigma_2 \in \mathbb{Q}\llangle a\rrangle^{n \times n}$, where
  $a = \poly{n}$.  Assume that either \(\vec{\mu_1} \neq \vec{\mu_2}\) or else
  \(\Sigma_1 \neq \Sigma_2\), and let
  \(D = \dtv \inp {\gaussian(\vct{\mu}_1,\Sigma_1),\gaussian(\vct{\mu}_2,\Sigma_2)} \).  Then,
  $D \geq \Delta = 2^{-\poly{n}}$.
\end{lemma}
\begin{proof}
  Let $X \sim \gaussian(\vct{\mu}_1, \Sigma_1)$ and
  $Y \sim \gaussian(\vct{\mu}_2, \Sigma_2)$.  Note that if \(\mu_1 = \mu_2\) and
  \(\Sigma_{1} = \Sigma_2\) then \(D = 0\).

Consider the first case that there exists $1\leq i \leq n$ such that
$\vct{\mu}_1(i) \neq \vct{\mu}_2(i)$ or $\Sigma_1(i,i) \neq \Sigma_2(i,i)$. The
TV-distance between $X,Y$ can be lower bounded by the TV-distance between $X_i$
and $Y_i$. Using the lower bound of \citet[Theorem 1.3; see
\cref{lem:lower-one}]{DMR23}, we have the following lower bound
\begin{align*}
    D &\geq \dtv\inp{X_i,Y_i}\\ &\geq \frac{1}{200}\min \left\{1, \max\left\{ \frac{|\Sigma_2(i,i)-\Sigma_1(i,i)|}{\Sigma_1(i,i)}, \frac{40 |\mu_1(i) - \mu_2(i)|}{\sqrt{\Sigma_1(i,i)}} \right\}\right\}\\ &\eqdef \Delta_1 \geq 2^{-\poly{n}},
\end{align*}
because each rational number is in $\mathbb{Q}\llangle a \rrangle$ with $a = \poly{n}$. 

Assume the first case does not hold, which implies $\vct{\mu}_1(i) = \vct{\mu}_2(i)$ and $\Sigma_1(i,i) = \Sigma_2(i,i)$ for all $1\leq i \leq n$.
There must exist $i\neq j$ such that $\Sigma_1(i,j) \neq \Sigma_2(i,j)$. Define
the vector $a \in \mathbb{R}^n$ such that $a_i=1,a_j = 1$ and $a_k = 0$ for all
$k \neq i,j$. By \cref{prop-bijection} , $D \geq \dtv \inp{a^TX, a^TY}$, where $a^T X \sim
\gaussian(\vct{\mu}_1(i)+\vct{\mu}_1(j), \Sigma_1(i,i)+2\Sigma_1(i,j)+\Sigma_1(j,j))$ and $a^T Y \sim
\gaussian(\vct{\mu}_2(i)+\vct{\mu}_2(j), \Sigma_2(i,i)+2\Sigma_2(i,j)+\Sigma_2(j,j))$. The same
proof as above shows that $D \geq \dtv \inp{a^TX, a^TY} \geq \Delta_2 = 2^{-\poly{n}}$.

Let $\Delta = \Delta_1$ for the first case and $\Delta = \Delta_2$ for the second case. 
\end{proof}

In passing, we note that several fundamental problems regarding the
computational complexity of diagonalization and related procedures remain open.
We refer to the recent survey by \cite{Sri23} for a discussion and to the papers
by \cite{ABBCS18}, \cite{BGKS23}, \cite{DeyKRS23}, and
\cite{shah_fast_2024,shahHermitianDiagonalizationLinear2025} for recent progress
and further references: this list of references is by no means meant to be
exhaustive.

\section{Proofs omitted from Section~\ref{sec:tv-dist-estim}}
\label{sec:proofs-omitted-from}

In this section, we often use the definition and standard properties of
conditional expectation, which may be found, e.g., in \citet[Sections 9.2 and
9.7]{Wil91}.

The following alternative expressions for the total variation distance are well
known, but we include the proofs for completeness.
\begin{fact}
  Let $P$ and $Q$ be probability measures on a measurable space $(\Omega,
  \F)$. Let the decomposition $P_Q, P_Q^{\perp}$ of $P$ and the random variable
  $R$ be as in \cref{def:radon-niko}.  Then, we have
  \begin{align}
    \dtv(P, Q) 
    &= \Ex_Q[(1-R)_+] \label{eq:tv-small}\\
    & = \Ex_Q[(R-1)_+] + P_Q^{\perp}(\Omega) \label{eq:tv-large}\\
    & =
      \frac{1}{2} \Ex_Q[\abs{1- R}] + \frac{1}{2}P_Q^{\perp}(\Omega).\label{eq:tv-both}
  \end{align}
  In particular, if $P \ll Q$ then
  $\dtv(P, Q) = \Ex_Q[(1-R)_+] = \Ex_Q[(R-1)_+] =
  \frac{1}{2}\Ex_Q[\abs{1-R}].$ (Here, we use the standard notation
  $x_{+} \defeq \max\inb{x, 0}$ for any real number $x$, so that
  $\abs{x} = x_+ + (-x)_+$.)
\end{fact}
\begin{proof}
  \Cref{eq:tv-both} follows by averaging \cref{eq:tv-small,eq:tv-large}, so we
  focus on proving the latter two equations.  Let $B \in \F$ be such that
  $Q(B^c) = 0$ and $P_Q^{\perp}(B) = 0$ (such a $B$ exists as $Q$ and
  $P_Q^{\perp}$ are mutually singular).  Denote by $C$ the subset
  $\inb{\omega \in \Omega \st R(\omega) \leq 1} \cap B$ of $B$ on which
  $R \leq 1$, and by $D$ the subset
  $\inb{\omega \in \Omega \st R(\omega) > 1} \cap B$ on which $R > 1$: both
  these sets are $\F$-measurable since $R$ is $\F$-measurable.  We then
  compute (using $Q(B^c) = P_Q^{\perp}(B) = 0$ and the decomposition in
  \cref{def:radon-niko}) that
  \begin{align}
    Q(C) - P(C) 
    &= \Ex_Q[(1-R)_+]\text{, and}\label{eq:2}\\
    P(D \cup B^{c}) - Q(D \cup B^c) 
    &= \Ex_Q[(R-1)_+] + P_Q^{\perp}(B^c) = \Ex_Q[(R-1)_+] + P_Q^{\perp}(\Omega).
  \end{align}
  Since $C^c = D \cup B^c$, this calculation also shows that the right hand
  sides of \cref{eq:tv-large,eq:tv-small} are equal.  We thus focus on proving
  \cref{eq:tv-small}.  To do this, it is sufficient to observe that for any
  $A \in \F$, we have $Q(A \cap C^c) - P(A \cap C^{c}) \leq 0$ and that
  $Q(A \cap C) - P(A \cap C) = \Ex_Q[(1-R)_+I_A]$. Adding these two equations
  and taking the supremum over all $A \in \F$ then shows that
  $\dtv(P, Q) \leq \Ex_Q[(1-R)_+]$, which combined with \cref{eq:2} completes
  the proof.
\end{proof}

\subsection{Proof of \texorpdfstring{\Cref{lem:tv-mon}}{Theorem Reference}}

\begin{proof}[Proof of \cref{lem:tv-mon}]
  \Cref{item:cR} follows from the definition of the TV functional since
  $x \mapsto (1-x)_+$ is a non-increasing function (and since $R \geq 0$).  For
  \cref{item:conditional}, it follows from standard properties of the
  conditional expectation~\citep[Section 9.7 (a, d)]{Wil91} that the conditional
  expectation of a valid ratio is also a valid ratio.  The first part of
  \cref{item:conditional} then follows from Jensen's inequality for conditional
  expectations~\citep[Section 9.7 (h)]{Wil91}, since the function
  $x \mapsto (1 - x)_+$ is convex:
  \begin{equation}
    \label{eq:5}
    \Ex[(1 - \Ex[R|\mathcal{G}])_+] \leq \Ex[\Ex[(1 - R)_+|\mathcal{G}]]
    = \Ex[(1 - R)_+].
  \end{equation}
  For the second part, consider the events
  $E_1 = \inb{R < 1 \text{ and } R_{\mathcal{G}} \geq 1}$ and
  $E_2 = \inb{R \geq 1 \text{ and } R_{\mathcal{G}} < 1}$.   We wish to show that
  $P(E_{1} \cup E_2) = 0$.  By the hypothesis for this item, both $E_1$ and $E_2$ are
  $\mathcal{G}$-measurable sets, so that (by the definition of the conditional
  expectation) the integral of the random variable $R - R_{\mathcal{G}}$ is zero
  on both $E_1$ and $E_2$.  However, this random variable is (by definition)
  strictly negative on $E_1$ and strictly positive on $E_2$.  It then follows by
  a standard argument (e.g., by considering probabilities of events of the form
  $E_{1} \cap \inb{R - R_{\mathcal{G}} < -1/n}$ for positive integers $n$ for
  the case of $E_1$) that $P(E_1) = P(E_2) = 0$, which implies the claim that
  the events $\inb{R < 1}$ and $\inb{R_{\mathcal{G}} < 1}$ are almost surely
  equivalent.  For the equality, we now compute using the above equivalence and
  the definition of the conditional expectation:
  \begin{equation*}
    \label{eq:9}
    TV(R) = \Ex[(1-R)\cdot I[R < 1]]
    = \Ex[(1-R)\cdot I[R_{\mathcal{G}} < 1]]
    = \Ex[(1-R_{\mathcal{G}})\cdot I[R_{\mathcal{G}} < 1]]
    = TV(R_{\mathcal{G}}). 
  \end{equation*}
  For the third part, we note that $\Ex[R \indep S|\sigma(R)] = \Ex[S]\cdot R$.
  Since $0 \leq \Ex[S] \leq 1$, we get the result by applying \cref{item:cR}
  followed by \cref{item:conditional}.
\end{proof}

\subsection{Proof of Remark in \texorpdfstring{\cref{def:discret}}{Theorem Reference}}

Towards the end of \cref{def:discret}, we remarked that for each interval
$J \in \cI_{\gamma,\delta}$, the events $\inb{R \in J}$ and
$\{\tilde{R} \in J\}$ are almost surely equivalent (i.e., the probability that
exactly one of them occurs is zero). Given the definition of $\tR$, this is
intuitively obvious, and the formal proof is very similar to the argument in the
proof of \cref{lem:tv-mon}.  We include this argument for completeness.

We show first that for any two distinct intervals
$J \neq J' \in \cI_{\gamma, \delta}$ the probability of the event
$\inb{R \in J, \tR \in J'}$ is zero.  Let $E$ denote this event.  By definition
of $\cI_{\gamma, \delta}$, any two distinct intervals in $\cI_{\gamma,\delta}$
are disjoint.  Thus, either all points in $J$ are smaller than all points in
$J'$, or else all points in $J'$ are smaller than all points in $J$.  Define
$Z = \tR - R$ in the former case and $Z = R - \tR$ in the latter case.  We then
have that $Z$ is \emph{strictly} positive when $E$ occurs.  In particular, if we
define $E_k \defeq E \cap \inb{Z > 1/k}$ for each positive integer $k$, then
$E = \bigcup_{k \geq 1}E_k$, so that if $\Pr{E_k} = 0$ for each positive integer
$k$ then $\Pr{E} = 0$ follows from the countable sub-additivity of
probabilities.  We now show that $\Pr{E_k} = 0$ for each positive integer $k$.

Note that by definition of the $\sigma$-field
$\mathcal{G} \defeq \sigma(\cI_{\gamma, \delta}(R))$, the event $\inb{R \in J}$
is $\mathcal{G}$-measurable (since $\inb{\R \in J}$ is by definition the same as
$\inb{\cI_{\gamma,\delta}(R) = J}$).  $\tR = \Ex[R | \mathcal{G}]$ is by
definition $\mathcal{G}$-measurable, and therefore so is the event
$\inb{\tR \in J'}$.  Thus, $E$ is also $\mathcal{G}$-measurable.  It thus
follows from the definition of conditional expectation that
$\Ex[R\cdot I[E]] = \Ex[\Ex[R|\mathcal{G}]\cdot I[E]] = \Ex[\tR\cdot I[E]]$.  We
thus get that $\Ex[Z\cdot I[E]] = 0$.  Thus, since $E_k \subseteq E$ and $Z$ is
strictly positive on $E$, we get that $\Ex[Z\cdot I[E_k]] = 0$ as well.  On the
other hand, when the event $E_k$ occurs, $Z > 1/k$.  Thus, we also have
$\Ex[Z\cdot I[E_k]] \geq \Pr{E_k}/k$.  We thus get that $\Pr{E_k} = 0$, for each
positive integer $k$.

We thus have $\Pr{R \in J, \tR \in J'} = 0$ whenever $J$ and $J'$
are distinct intervals in $\cI_{\gamma,\delta}$.  The claim in the remark now
follows since $R, \tR$ are non-negative so that for any interval
$J \in \cI_{\gamma,\delta}$,
\begin{align*}
  \Pr{R \in J, \tR \not\in J} 
  &= \sum_{\substack{J' \in \cI_{\gamma,\delta}\\J'
  \neq J}} \Pr{R \in J, \tR \in J'} = 0 \text{, and}\\
  \Pr{R \not\in J, \tR \in J} 
  &= \sum_{\substack{J' \in \cI_{\gamma,\delta}\\J'
  \neq J}} \Pr{\tR \in J, R \in J'} = 0.
\end{align*}

\subsection{Properties of the \texorpdfstring{$(\gamma,\delta)$}{Theorem Reference}-Discretization}

The following elementary property of $(\gamma,\delta)$-partitions turns out to
be very useful.
\begin{observation}\label{obv-elem-discrete}
  Suppose $x, y \in [0, 1]$ lie in the same partition in $\cI_{\gamma, \delta}$,
  i.e., $\cI_{\gamma,\delta}(x) = \cI_{\gamma,\delta}(y)$.  Then 
  \begin{equation}
    \label{eq:7}
    \abs{x - y} \leq
    \begin{cases}
      \gamma & \text{if $x, y \in I_1$.}\\
      (1 - x) \delta & \text{otherwise}.
    \end{cases}
  \end{equation}
\end{observation}
\begin{proof}
  We use the notation of \cref{def:partition}.  If $x, y \in I_0 = \inb{1}$,
  there is nothing to prove.  Otherwise, if $x, y \in I_k$ for $k \geq 1$, we
  have $\abs{x -y} \leq a_{k-1} - a_k$, which is equal to $\gamma$ when $k = 1$.
  Otherwise, when $k > 1$, it equals $(1- a_{k-1})\delta$, which is at most
  $(1 - x )\delta$.
\end{proof}

We record the following simple properties of the
$(\gamma,\delta)$-discretization.
\begin{observation}\label{obv:discrete}
  Let $\tR$ be the $(\gamma,\delta)$-discretization of a valid ratio $R$ defined
  on some probability space $(\Omega, \F, P)$.  Then
  \begin{enumerate}
  \item \label{item:discrete-low} $\Ex[\abs{R - \tR}I[R < 1]] \leq \gamma +
    \delta TV(R)$.
  \item \label{item:discrete-high}
    $\Ex\insq{\abs{\frac{1}{R} - \frac{1}{\tR}}RI[R \geq 1]} \leq \gamma +
    \delta TV(R)$.
  \end{enumerate}
\end{observation}
\begin{proof}
  Let the intervals $I_0, I_1, J_1, \dots, I_m, J_m$ be as in the definition of
  $\cI_{\gamma,\delta}$ in \cref{def:partition}.  Then
  $I[R < 1] = \sum_{i=1}^mI[R \in I_i]$ and
  $I[R \geq 1] = I[R \in I_0] + \sum_{i=1}^mI[R \in J_i].$ As noted in
  \cref{def:discret}, the random variables $I[R \in I]$ and $I[\tR \in I]$ are
  almost surely equal.  Thus, almost surely, $R \in I \implies \tR \in I$ for
  each $I \in \cI_{\gamma,\delta}$.  Thus, by the construction of
  $\cI_{\gamma,\delta}$ (see \cref{obv-elem-discrete}), we have
  $\abs{R - \tR} \leq \gamma$ when $R \in I_1$, and
  $\abs{R - \tR} \leq a_{k-1} - a_k = \gamma\delta(1+\delta)^{k-2} =
  \delta(1-a_{k-1}) \leq \delta(1 - R)_{+}$ when $R\in I_k$ for
  $2 \leq k \leq m$.  Combining these gives \cref{item:discrete-low}.
  Similarly, we have $R = \tR$ when $R \in I_0$,
  $\abs{1/R - 1/\tR}R \leq \gamma R$ when $R \in J_1$ and
  $\abs{1/R - 1/\tR}R \leq \delta(R -1)_+$ when $R \in J_k$ for
  $2 \leq k \leq m$.  Combining these and recalling that $\Ex[R] \leq 1$ gives
  \cref{item:discrete-high}.
\end{proof}

\subsection{Proof of \texorpdfstring{\cref{thm-ext-estimate}}{Theorem Reference}}
\begin{proof}[Proof of \cref{thm-ext-estimate}]
  Consider any valid ratio $S$ defined on some probability space
  $(\Omega', \F', P')$.  Then the definitions of $S\indep R$ and $S \indep \tR$
  show that we can assume that $S, R, \tR, S\indep R$ and $S \indep \tR$ are
  all defined on the space
  $(\Omega' \times \Omega, \F' \times \F, P' \times P)$, in such a way that $S$
  is independent of $R$ and $\tR$.  We now note that
  \begin{equation}
    \label{eq:13}
    \begin{aligned}
      \abs{TV(S \indep R) - TV(S \indep \tR)} 
      &= \frac{1}{2}\cdot
        \Big|\Ex\insq{
        \abs{1 - SR} - \abss{1 - S\tR}
        } + \Ex[S(\tR - R)]\Big|\\
      &=\frac{1}{2}\cdot
        \Big|\Ex\insq{
        (\abs{1 - SR} - \abss{1 - S\tR})I[R < 1]
        } \\
      &\qquad\qquad + {\Ex\insq{
        (\abs{1 - SR} - \abss{1 - S\tR})I[R \geq 1]
        }} \\
      &\qquad\qquad + \Ex[S(\tR - R)]\Big|.
    \end{aligned}
  \end{equation}
  Note that the last term is zero since by the independence of $S$ with $R$ and
  $\tR$,
  \begin{equation}
    \label{eq:19}
    \Ex[S(\tR - R)] = \Ex[S]\cdot\Ex[\tR - R] = 0,
  \end{equation}
  since $\tR$ is a conditional expectation of $R$.  For the first term we have
  \begin{equation}
    \label{eq:12}
    \begin{aligned}
      \abs{\Ex\insq{
      (\abs{1 - SR} - \abss{1 - S\tR})I[R < 1]
      }}
      &\leq
        \Ex\insq{
        \abs{\abs{1 - SR} - \abss{1 - S\tR}}\cdot I[R < 1]
        }\\
      &\leq \Ex[S\cdot \abss{R - \tR} \cdot I[R < 1]]\\
      &= \Ex[S] \cdot \Ex[\abss{R - \tR}\cdot I[R < 1]] \leq \gamma + \delta TV(R).
    \end{aligned}
  \end{equation}
  Here, the first inequality is Jensen's inequality, the second inequality uses
  $\abs{\abs{a} - \abs{b}} \leq \abs{a - b}$ for real $a$ and $b$, the equality
  is by the independence of $S$ from $R$ and $\tR$, and the last inequality is
  from \cref{item:discrete-low} of \cref{obv:discrete} and the fact that $S$
  being a valid ratio satisfies $\Ex[S] \leq 1$.

  For the second term, we begin by noting that for any \emph{fixed}
  $\omega' \in \Omega'$, the random variable
  $\abs{\frac{1}{\tR} - S(\omega')}\cdot I[R \geq 1]$ is a random variable on
  the probability space $(\Omega, \F, P)$ that is measurable with respect to the
  sub $\sigma$-algebra $\mathcal{G} \defeq \sigma(\cI_{\gamma,\delta}(R))$
  (since so are $\tR$ and $I[R \geq 1]$).  We then have
  \begin{align}
    \label{eq:15}
    \int\limits_{\omega \in \Omega}{\abs{\frac{1}{\tR} - S(\omega')}\cdot
    I[R \geq 1]\cdot R} \, P(d\omega)
    &= \int\limits_{\omega \in \Omega}\Ex\insq{{\abs{\frac{1}{\tR} - S(\omega')}\cdot
      I[R \geq 1]\cdot R}\Bigg|\mathcal{G}} \, P(d\omega)\\
    = \int\limits_{\omega \in \Omega}\abs{\frac{1}{\tR} - S(\omega')}\cdot
      I[R \geq 1]\cdot \Ex[R|\mathcal{G}] \, P(d\omega)
    &= \int\limits_{\omega \in \Omega}\abs{\frac{1}{\tR} - S(\omega')}\cdot
      I[R \geq 1]\cdot \tR \, P(d\omega),
  \end{align}
  where the first and second equalities are by standard properties of the
  conditional expectation~\citep[Section 9.7 (a, j)]{Wil91}, while the third is
  by the definition of $\tR$ as $E[R|\mathcal{G}]$.  Integrating this equation
  over $\omega'$ using the measure $P'$, we then have from Fubini's theorem that
  \begin{equation}
    \label{eq:16}
    \Ex\insq{\abs{\frac{1}{\tR} - S}\cdot
      I[R \geq 1]\cdot R}
    =
    \Ex\insq{\abs{\frac{1}{\tR} - S}\cdot
      I[R \geq 1]\cdot \tR}.
  \end{equation}
  We can now estimate the second term on the RHS in \cref{eq:13} as
  follows.
  \begin{equation}
    \label{eq:17}
    \begin{aligned}
      \abs{
      \Ex\insq{
      (\abs{1 - SR} - \abss{1 - S\tR})
      I[R \geq 1]
      }}
      &= \abs{\Ex\insq{
        \inp{\abs{\frac{1}{R} - S} R  - \abs{\frac{1}\tR - S}\tR}
        \cdot I[R \geq 1]
        }}\\
      &= \abs{\Ex\insq{
        \inp{\abs{\frac{1}{R} - S} R - \abs{\frac{1}\tR - S}R} \cdot I[R
        \geq 1]}}\\
      &\leq \Ex\insq{
        \abs{\abs{\frac{1}{R} - S}  - \abs{\frac{1}\tR - S}}\cdot R \cdot I[R
        \geq 1]}\\
      &\leq \Ex\insq{
        \abs{\frac{1}{R} - \frac{1}{\tR}} \cdot R \cdot I[R
        \geq 1]} \leq \gamma + \delta TV(R).
    \end{aligned}
  \end{equation}
  Here, the first equality is valid since on the event $\inb{R \geq 1}$ we also
  have $\inb{\tR \geq 1}$ so that $1/R$ and $1/\tR$ are bounded random variables
  on this event.  The second equality uses \cref{eq:16}.  The two inequalities
  after that follow from Jensen's inequality and the triangle inequality,
  exactly as in the estimation of the first term of \cref{eq:13} as carried out
  in \cref{eq:12}. The final inequality uses \cref{item:discrete-high} of
  \cref{obv:discrete}.  The claim now follows by substituting the estimates in
  \cref{eq:19,eq:12,eq:17} into \cref{eq:13}, and then taking supremum over all
  possible choices of $S$.
\end{proof}

\subsection{Omitted Arguments in the Proof of \texorpdfstring{\cref{thm-alg-dis-prod-correct}}{Theorem Reference}}
In the proof of \cref{thm-alg-dis-prod-correct}, we noted that a direct
inductive argument shows that each $Y_i$ is a conditional expectation of
$R_1 \indep R_2 \indep \dots \indep R_i$ (where the notation is as in the proof
of \cref{thm-alg-dis-prod-correct}).  While this is an easy and standard
argument, we include the details here for completeness.

\newcommand{\cF}{\ensuremath{\mathcal{F}}} \newcommand{\cH}{\ensuremath{\mathcal{H}}} \newcommand{\cG}{\ensuremath{\mathcal{G}}} All the following arguments are performed on the probability space
$([M]^n, \mathcal{H}, Q)$, where $\mathcal{H}$ is the $\sigma$-field of all
subsets of $[M]^n$, and $Q$ is the product distribution on this space defined in
\cref{alg-dis-prod-tv}.  For $1 \leq i \leq n$, let $\cF_i$ denote the $\sigma$-field
$\sigma(R_i)$ generated by $R_i$, and define the product $\sigma$-fields
$\cH_i \defeq \cF_1 \times \dots \times \cF_i$.  Under the underlying product
probability distribution $Q$, note that for all $j > i$, any event in $\cH_i$ is
independent of any event in $\cF_j$.\footnote{Note that since
  \cref{thm-alg-dis-prod-correct} is in the setting of finite sets, the
  ``$\sigma$-fields'' here are just finite set systems closed under complement
  and union, and containing the universal set.  However, it is convenient to do
  the argument in this general setting. Note, however, that $\times$ denotes the
  product operation on $\sigma$-fields, and not merely the Cartesian product:
  for example, $\cF_1 \times \cF_2$ is the set system generated by sets of the
  form $\inb{A \times B \st A \in \cF_1 \text{ and } B \in \cF_2}$.}  We will
show by induction that for every $1 \leq i \leq n$, there is a sub
$\sigma$-field $\cG_i \subseteq \cH_i$ such that
$Y_{i} = \Ex[R_{1}R_2\dots R_{i} | \cG_i]$.

The claim is trivially true in the base case $i = 1$, since $Y_1 = R_1$, so that
we can take $\cG_1 = \cF_1 = \cH_1$.  For the inductive case, suppose that
$Y_{i} = \Ex[R_{1}R_2\dots R_{i} | \cG_i]$ for some $\sigma$-field
$\cG_i \subseteq \cH_i$.  Note that this implies that $Y_i$ is $\cG_i$ measurable
(in other words, $\sigma(Y_i) \subseteq \cG_i$).  Thus, the $\sigma$-field
$\cG'_i \defeq \sigma(\cI_{\gamma,\delta}(Y_i)) \subseteq \sigma(Y_i)$ is a
sub-$\sigma$-field of $\cG_i$ as well.  Since $\tilde{Y}_i = \Ex[Y_i|\cG'_i]$,
we thus see from the tower property of conditional expectations~\citep[Section
9.7 (i)]{Wil91} that
\begin{equation}
  \label{eq:11}
  \tilde{Y}_i = \Ex\big[\Ex[R_{1}R_2\dots R_{i} | \cG_i]|\cG'_i\big] =
  \Ex[R_{1}R_2\dots R_{i} | \cG'_i],
\end{equation}
where $\cG'_i \subseteq \cH_{i}$.  Now, define
$\cG_{i+1} \defeq \cG'_i \times \F_{i+1}$, which by definition, is a
sub-$\sigma$-field of $\cH_{i+1}$.  Note also that $R_{i+1}$ is measurable with
respect to $\cF_{i+1}$ (by definition) and hence also with respect to
$\cG_{i+1}$.  Thus, by a standard property of conditional
expectations~\citep[Section 9 (j)]{Wil91},
\begin{equation}
  \label{eq:14}
  \Ex[R_1R_2\dots R_iR_{i+1}|\cG_{i+1}] =\Ex[R_1R_2\dots R_i|\cG_{i+1}] R_{i+1}.
\end{equation}
Now, as noted above, by construction of the product distribution $Q$, every
event in $\cG'_i \subseteq \cH_i$ is independent of every event in $\cF_{i+1}$.
It then follows from the construction of the product $\sigma$-field and the
interaction between independence and conditional expectation~\citep[Section
9.7 (k)]{Wil91} that
\begin{equation}
  \label{eq:21}
  \tilde{Y}_i \overset{\textup{\cref{eq:11}}}{=}
  \Ex[R_{1}R_2\dots R_{i} | \cG'_i] = \Ex[R_{1}R_2\dots R_{i} | \cG'_i \times \F_{i+1}].
\end{equation}
Combining \cref{eq:14,eq:21}, we thus see that $Y_{i+1}=\tilde{Y}_iR_{i+1}$
satisfies
\begin{equation}
  \label{eq:26}
  Y_{i+1} = \Ex[R_1R_2\dots R_iR_{i+1}|\cG_{i+1}],
\end{equation}
which completes the induction.

\subsection{Proof of \texorpdfstring{\cref{lem-prod-discrete}}{Theorem Reference}}

\begin{proof}[Proof of \cref{lem-prod-discrete}]
  Define $Z_i = \tR_1 \indep \tR_2 \indep \dots \indep \tR_{i-1}$, and
  $S_i = R_{i+1} \indep R_{i+2} \indep \dots \indep R_n$ (so that
  $Z_1 = S_n = 1$).  Then we have
  \begin{multline}
    \label{eq:18}
    \abs{TV(R_1 \indep R_2 \indep \dots \indep R_n) - TV(\tR_1 \indep \tR_1 \indep
      \dots \indep \tR_{n})} = \\
    \leq \sum_{i=1}^n \abs{TV(Z_i \indep S_i \indep
      R_i) - TV(Z_i \indep S_i \indep \tR_i)} \leq \sum_{i=1}^n \ext(R_i, \tR_i),
  \end{multline}
  where the last inequality is by the definition of the extension distance.  The
  claim now follows from \cref{thm-ext-estimate}.  (Note that we also use
  implicitly the fact that the $\indep$ operation is ``commutative'' in the
  sense that $R_1 \indep R_2$ has the same law as $R_2 \indep R_1$.)
\end{proof}

\section{Proofs omitted from Section\texorpdfstring{~\ref{sec:alg}}{Theorem Reference}}
\label{sec:proofs-omitted-from-1}
We need to use the following result of~\cite{DMR23}.
\begin{lemma}[\textbf{\citet[Theorem 1.3]{DMR23}}]\label{lem:lower-one}
Let $X \sim \gaussian(\mu,\sigma^2)$ and $Y \sim \gaussian(0,1)$ be two one-dimensional Gaussians, then 
\begin{align*}
\frac{1}{200} \min \left\{ 1, \max\{|1-\sigma^2|, 40|\mu|\} \right\}   \leq  \dtv \inp {X,Y} \leq \frac{3|\sigma^2-1|}{2} + \frac{|\mu|}{2}.
\end{align*}
\end{lemma}
\subsection{Proof of Lemma~\ref{lem:lower}}\label{app:lower}
\begin{proof}[Proof of \Cref{lem:lower}]
Using \Cref{lem:lower-one}, we know that $\Delta_i \leq \dtv \inp {X_i,Y_i}$. Furthermore, if $\max\{|1-\Sigma_{i,i}|, 40|\mu_i|\} \geq 1$, then $\Delta_i = 1/200$ and $10^4 \Delta_i > 1 \geq \dtv \inp{X_i,Y_i}$. If $\max\{|1-\Sigma_{i,i}|, 40|\mu_i|\}< 1$, we have
\begin{align*}
  10^4 \Delta_i \geq  \frac{25}{2} |1-\Sigma_{i,i}| + \frac{25}{2}\times 40|\mu_i| \geq \frac{3}{2}|1-\Sigma_{i,i}| + \frac{1}{2}|\mu_i| \geq \dtv \inp{X_i,Y_i}. 
\end{align*}
Hence, for any $1 \leq i \leq n$, we have 
\begin{align*}
    \Delta_i \leq \dtv \inp{X_i,Y_i} \leq 10^4\Delta_i.
\end{align*}
Recall $\Delta = \max_i \Delta_i$. We have $\Delta \leq \max_{i}\dtv
\inp{X_i,Y_i} \leq \dtv \inp{X,Y}$. In the other direction,  $ n 10^4\Delta \geq
10^4\sum_{i=1}^n\Delta_i \geq \sum_{i=1}^n\dtv \inp{X_i,Y_i} \geq
\dtv\inp{X,Y}$, where the last inequality holds due to $\sum_{i=1}^n \dtv
\inp{X_i,Y_i} \geq  \dtv \inp {X,Y}$. To verify the last inequality, consider a
coupling of $X,Y$ that couples each dimensional independently and optimally, and
the inequality follows from the coupling characterization of the total variation
distance.
\end{proof}

\end{document}